\documentclass[11pt, onecolumn, conference]{IEEEtran}
\usepackage{cite}
\usepackage[cmex10]{amsmath}
\usepackage{array}
\usepackage{verbatim}
\usepackage{graphicx} %%For loading graphic files
\usepackage{pdfsync}
\usepackage{lineno}
\usepackage[tight,footnotesize]{subfigure}
\usepackage{amsthm}
\usepackage{hyperref}

\newtheorem{theorem}{Theorem}

\newtheorem{lemma}{Lemma}

\newtheorem*{repos*}{Reducibility Postulate}

\begin{document}

\title{Solvability of Cubic Graphs - \\
From Four Color Theorem to NP-Complete}

%\thanks{Tony T. Lee was supported by the National Science Foundation of China under Grant 61172065. Qingqi Shi was supported by the Hong Kong RGC Earmarked Grant CUHK414012.}

\author{
	\IEEEauthorblockN{ Tony T. Lee \IEEEauthorrefmark{1} and Qingqi Shi\IEEEauthorrefmark{2}}
	\IEEEauthorblockA{\IEEEauthorrefmark{1}State Key Laboratory of Advanced Communication Systems and Networks, \\
	Department of Electronic Engineering, Shanghai Jiao Tong University,
	Shanghai, China\\
	ttlee@ie.cuhk.edu.hk}
	\IEEEauthorblockA{\IEEEauthorrefmark{2}Department of Information Engineering,
	The Chinese University of Hong Kong,
	Hong Kong, China\\
	qqshi@ie.cuhk.edu.hk}
	
}

\maketitle

\begin{abstract}

Similar to Euclidean geometry, graph theory is a science that studies figures that consist of points and lines. The core of Euclidean geometry is the parallel postulate, which provides the basis of the geometric invariant that the sum of the angles in every triangle equals $\pi$ and Cramer's rule for solving simultaneous linear equations. Since the counterpart of parallel postulate in graph theory is not known, which could be the reason that two similar problems in graph theory, namely the four color theorem (a topological invariant) and the solvability of NP-complete problems (discrete simultaneous equations), remain open to date. In this paper, based on the complex coloring of cubic graphs, we propose the reducibility postulate of the Petersen configuration to fill this gap. Comparing edge coloring with a system of linear equations, we found that the postulate of reducibility in graph theory and the parallel postulate in Euclidean geometry share some common characteristics of the plane. First, they both provide solvability conditions on two equations in the plane. Second, the two basic invariants of the plane, namely the \textit{chromatic index} of bridgeless cubic plane graphs and the \textit{sum of the angles} in every triangle, can be respectively deduced from them in a straightforward manner. This reducibility postulation has been verified by more than one hundred thousand instances of Peterson configurations generated by computer. Despite that, we still don't have a logical proof of this assertion. Similar to that of the parallel postulate, we tend to think that describing these natural laws by even more elementary properties of the plane is inconceivable.
\end{abstract}


\section{Introduction and Overview}
\label{sec1}
Similar to Euclidean geometry, graph theory is a science that studies figures that consist of points and lines. Instead of measuring angles and distances, graph theory focuses on the topological configurations that are composed of vertices and edges. The core of Euclidean geometry is the fifth postulate, commonly called the parallel postulate. To the ancients, however, the parallel postulate was less obvious than the other four postulates. For last two thousand years, many tried in vain to prove the parallel postulate using Euclid's other four postulates \cite{greenberg2008euclidean}. Some false proofs of the parallel postulate were accepted for many years before they were exposed. It is now known that the parallel postulate is a natural law of two dimensional Euclidean planes, and a proof is impossible. This law has produced the following monumental ramifications: 
\begin{enumerate}
\item  The geometric invariant that the sum of the angles in every triangle equals $\pi$ is a direct consequence of the parallel postulate. A generalization of this result is the Gauss-Bonnet theorem in differential geometry.
\item  The parallel postulate provides the solvability condition of two linear equations in the plane. The theory of determinant and Cramer's rule for solving simultaneous linear equations are generalization of this condition in Euclidean space. 
\end{enumerate}

Since the counterpart of the parallel postulate in graph theory is not known, which could be the reason that the theoretical proofs or solutions of two similar problems in graph theory, namely the four color theorem (a topological invariant) and the solvability of NP-complete problems (discrete simultaneous equations), remain open to date. In this paper, based on the complex coloring method described in \cite{Lee2013}, we propose the reducibility postulate of the Petersen configuration to fill this gap. An immediate consequence of this proposition is the 3-edge coloring theorem, or equivalently, the four color theorem (4CT). This self-evident proposition has been verified by more than one hundred thousand instances generated by computer. Despite that, we still don't have a logical proof of this assertion.

The 4CT is the holy grail of graph theory, but the proof of this simply stated theorem is elusive. The 4CT states that the minimum number of colors required to color a map is four, which represents a topological invariant of the plane. Ever since the problem was raised by Francis Guthrie in 1852, the theorem was falsely proved twice by Alfred Kempe \cite{kempe1879} in 1879, and Peter Tait \cite{Tait1880} in 1880. Despite their fruitless efforts, their ideas provided fundamental insights into graph coloring that are still of paramount importance in graph theory. The 4CT was finally proved by Kenneth Appel and Wolfgang Haken \cite{Appel1977_1,Appel1977_2} in 1976. Their proof relies on computer-aided checking that cannot be verified by a human. In 1997, Neil Robertson, Daniel Sanders, Paul Seymour, and Robin Thomas published a simpler version of a computer-assisted proof based on the same idea \cite{Robertson19972}. The history of 4CT is detailed in \cite{Fritsch1998four,wilson2002four}, and a brief survey on the progress is provided in \cite{Thomas98anupdate}.

In his 1880 paper, Peter Tait proposed that any bridgeless cubic planar graph has a Hamiltonian cycle. His proof of 4CT based on this false assumption was refuted by Julius Petersen \cite{Petersen1891} in 1891. However, it was not until 1946 that William Tutte found that there are such planar graphs without any Hamiltonian tours. The main contribution of his paper was to establish the following equivalent formulation of 4CT.

\begin{theorem}[Tait]
A bridgeless cubic planar graph $G$ is 4-face-colorable if and only if $G$ is 3-edge-colorable.
\end{theorem}
A proof of this theorem can be found in many books and papers \cite{chartrand2011graphs, west2001introduction}. It essentially transforms the 4CT from a vertex-coloring problem into an edge-coloring problem.

In computer science, satisfiability, abbreviated as SAT, is the problem of determining if there exists a truth assignment of variables that satisfies a given Boolean formula. SAT was the first known example of an NP-complete problem. The satisfiability problem of a Boolean expression $\varphi$ can also be considered as the solvability of a set of simultaneous Boolean equations \cite{papadimitriou1995computational}. For example, a truth assignment of the expression $\varphi=(x_1\vee x_2\vee x_3)\wedge(\overline{x_2}\vee x_3)\wedge (x_1\vee \overline{x_3})$ is a solution of the following set of Boolean equations:

\begin{eqnarray}
 x_1 \vee x_2 \vee x_3 = 1,      \nonumber \\
   \overline{x}_2 \vee x_3=1, \nonumber \\
\notag   x_1 \vee \overline{x_3}=1.
\end{eqnarray}

Many NP-complete problems can also be considered as solving a set of simultaneous discrete equations. One of them is to decide the chromatic index of a cubic graph. The chromatic index $\chi_{e}(G)$ of a simple graph $G$ is the minimum number of colors required to color the edges of the graph such that no adjacent edges have the same color. A theorem proved by Vizing \cite{vizing1964} states that the chromatic index is either $\Delta$ or $\Delta+1$, where $\Delta$ is the maximum vertex degree of graph $G$. Graph $G$ is said to be Class 1 if $\chi_{e}(G)=\Delta$; otherwise, it is Class 2. According to Tait's equivalent formulation, the 4CT is established if every bridgeless cubic planar graph $G$ is Class 1. In \cite{Holyer1981}, Holyer proved that a Boolean expression $\varphi$ can be converted into a cubic graph $G$, such that $\varphi$ is satisfiable if and only if $G$ is 3-edge colorable. Therefore, it is NP-complete to determine the chromatic index of an arbitrary cubic graph.

The concept of variable edges is introduced in the complex coloring method; a proper coloring is achieved by eliminating all variables in a color configuration of a cubic graph. A connected bridgeless cubic graph that does not have 3-edge coloring is called a snark. Many important and difficult problems in graph theory are related to snarks. However, their properties and structures are still largely unknown. For the first time, this paper completely specifies the necessary and sufficient condition of snarks in terms of complex coloring configurations. In principle, the entire graph $G=(V,E)$ can be considered as a set of simultaneous equations, in which each vertex $v\in V$ represents a constraint on coloring of edges. The variable elimination procedure of edge coloring is similar to the algebraic method for solving systems of linear equations. A comparison between these two procedures is summarized in Table \ref{table1}.

\begin{table}[ht]
\label{table1}
\centering
\caption{A comparison between linear equations and edge coloring}
\begin{tabular}{| c | c | c |} 
\hline
 & \textbf{System of linear equations} \newline
(Euclidean Geometry)
 & \textbf{Edge coloring} \newline
(Graph Theory)
 \\
\hline
\textbf{Operations} & Arithmetic Operations & Color Exchanges  \\ 
\hline
\textbf{Constraints} & Linear Equations & Vertices and edges  \\
\hline
\textbf{Unknowns} & Variables & Variable-edges \\
\hline
\textbf{Algorithms} & Variable Elimination & Variable Elimination \\
\hline
\textbf{Solution} & Consistency & 3-edge coloring \\ 
\hline
\textbf{No solution} & Inconsistency & Snark \\  
\hline
\end{tabular}
\end{table}

The inconsistency of two linear equations is usually interpreted as two parallel lines in a Euclidean space. In the complex coloring of the Petersen graph, the smallest snark, the configuration contains two variables in two disjoint odd cycles, which can never be eliminated because they will never meet each other. That is, the two cycles behave the same as \textit{two parallel lines} in a Euclidean space. The analogy between the parallel postulate in Euclidean geometry and the reducibility postulate in graph theory is illustrated in Fig.~\ref{fig_1}. They share some common characteristics of the plane. First, they can respectively deduce the two invariants of the plane. Second, they both provide solvability conditions on equations in the plane. 

\begin{figure}[htbp]
\centering
	\subfigure[Euclidean geometry.]{
		\includegraphics[scale=0.7]{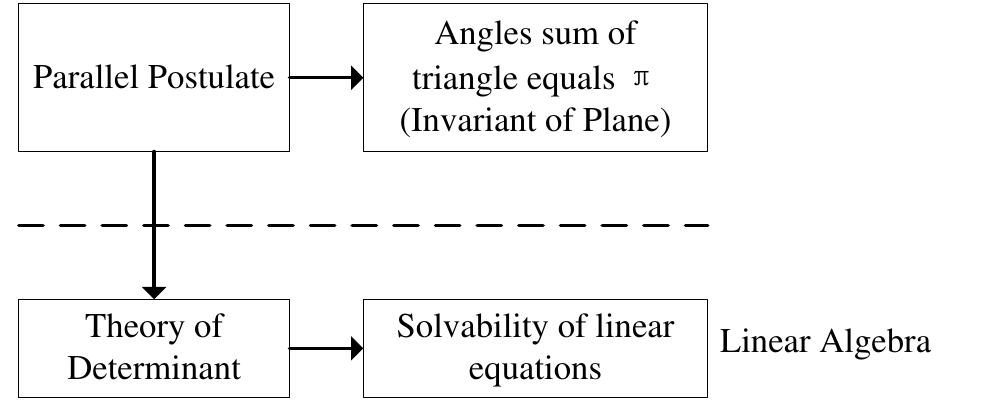}
		\label{fig_1_a}
		} \quad
	\subfigure[Graph theory.]{
		\includegraphics[scale=0.7]{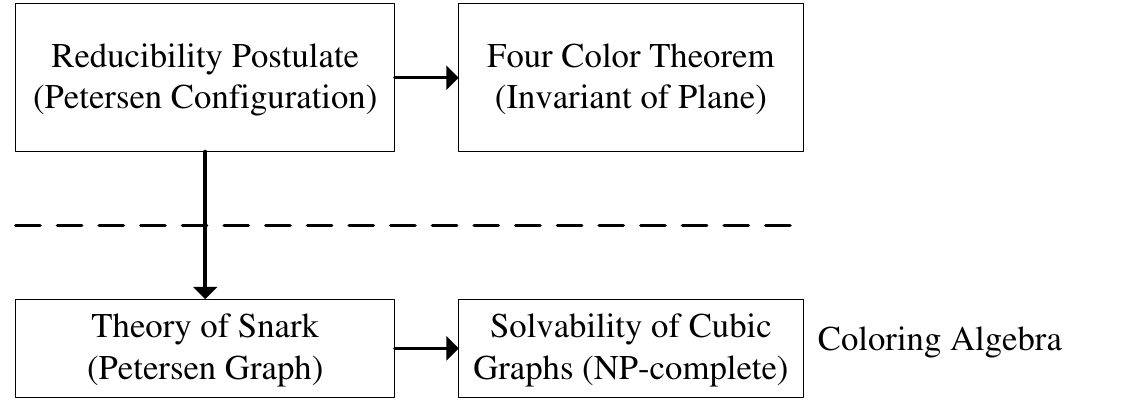}
		\label{fig_1_b}
		}
\caption{The two postulates and their ramifications.} 
\label{fig_1}
\end{figure}

The rest of this paper is organized as follows. In section \ref{sec2}, we briefly describe the basic concept of complex colors and the rules of color exchanges. In section \ref{sec3}, we introduce the decomposition of a three-colored configuration into two-colored maximal sub-graphs. In section \ref{sec4}, we define the reducibility of configurations, from which we show that a snark can be characterized by a closed set of irreducible configurations. In Section \ref{sec5}, we introduce the Petersen configuration, and propose the postulate that a Petersen configuration of a bridgeless cubic planar graph must be reducible. In section \ref{sec6}, we show that the 3-edge coloring theorem of bridgeless cubic planar graphs is an immediate consequence of the reducibility postulate. The discussion in Section \ref{sec7} focuses on a comparison between graph theory and Euclidean geometry, in particular, the analogy between the reducibility postulate in graph theory and the parallel postulate in Euclidean geometry. Section \ref{sec8} provides a conclusion. Furthermore, Appendix \ref{appdx_A} provides examples of contraction of snarks to the Petersen graph. A video clip to demonstrate the operation of complex coloring was posted at YouTube at \url{http://www.youtube.com/watch?v=KMnj4UMYl7k}.

\section{Preliminaries}
\label{sec2}
An edge-coloring method based on exchanges of complex colors is proposed in \cite{Lee2013}. The basic idea is to partition each edge of a graph into two links, then color the links and perform color exchanges between the links instead of the edges. This color exchange method is briefly described in this section to facilitate our discussions. Even though the complex coloring method can be applied to any graph, we only concentrate on the complex coloring of bridgeless cubic graphs, because it is the focal point of this paper.

\subsection{Complex Coloring of Cubic Graphs}
Let $G=(V,E)$ be a cubic graph with vertex set $V$, edge set $E$. The \textbf{\textit{incidence graph}} $G^*$ is constructed from $G$ by placing a fictitious vertex in the middle of each edge of $G$. Let $E^*(G^*)=\{e_{i,j}^*|e_{i,j}\in E(G) \}$ denote the set of fictitious vertices on edges. Then edge $e_{i,j}\in E(G^* )$ consists of two \textbf{\textit{links}}, denoted by $l_{i,j}=(v_i,e_{i,j}^*)$ and $l_{j,i}=(v_j,e_{i,j}^*)$, which connect two end vertices $v_i$ and $v_j$  of $e_{i,j}$. Fig.~\ref{fig_2_a} illustrates the incidence graph of the tetrahedron.

Let $L(G^*)$ be the set of links and $C=\{a, b, c\}$ denote the set of three colors. A \textbf{\textit{coloring function}} is a mapping of colors on links, $\sigma : L(G^*)\rightarrow C$. The color of link $l_{i,j}\in L(G^*)$ is denoted as $\sigma(l_{i,j})=c_{i,j}$. Since each edge $e\in E(G^*)$ consists of two links, the color function can also be considered as a mapping defined on the set of edges, $\sigma:E(G^*)\rightarrow C\times C$. Since the coloring function assigns two colors to each edge of graph $G$, or one color for each link of the incidence graph $G^*$, the mapping $\sigma$ is called a \textbf{\textit{complex coloring}} of graph $G$.

The coloring function $\sigma$ is \textbf{\textit{consistent}} if colors assigned to those links incident to the same vertex $v$ are all distinct for all $v\in V(G)$. We define the \textbf{\textit{colored edge}} $\sigma(e_{i,j})={\vec e}_{i,j}=(c_{i,j},c_{j,i})=(\alpha, \beta)$, $\alpha, \beta \in C$ as a two-tuple color vector, where $c_{i,j}=\sigma(l_{i,j})=\alpha$ and $c_{j,i}=\sigma(l_{j,i})=\beta$ are respective colors of the two links of $e_{i,j}$. The vector of colored edge $e_{i,j}=(\alpha, \beta)$ is a \textbf{\textit{variable}} if $\alpha \neq \beta$; otherwise, $e_{i,j}=(\alpha, \alpha)$ is a \textbf{\textit{constant}}. A proper 3-edge coloring of graph $G$, called a Class 1 graph, can be achieved by eliminating all variables. The graph $G$ is a Class 2 graph if variables cannot be eliminated.

As an example, a consistent coloring of a tetrahedron containing two $(a, b)$ variables is shown in Fig.~\ref{fig_2_b}, and a proper coloring of a tetrahedron is shown in Fig.~\ref{fig_2_c} with the set of colors $C=\{a, b, c\}$, where $a$, $b$, and $c$ represent red, green, and blue colors, respectively.

\begin{figure}[htbp]
 \centering
	\subfigure[Incidence graph of tetrahedron.]{
		\includegraphics[scale=1.0]{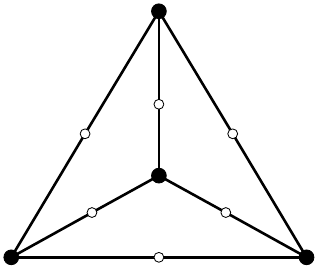}
		\label{fig_2_a}
		} 
 \quad
	\subfigure[Consistent coloring of tetrahedron.]{
		\includegraphics[scale=1.0]{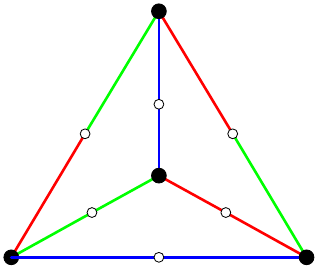}
		\label{fig_2_b}
		}
 \quad
	\subfigure[Proper coloring of tetrahedron.]{
		\includegraphics[scale=1.0]{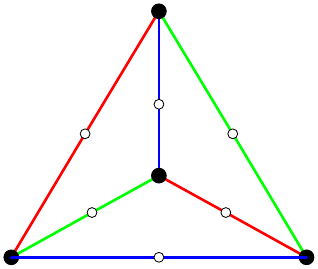}
		\label{fig_2_c}
		}
\caption{Complex coloring of the tetrahedron graph.} 
\label{fig_2}
\end{figure}

\subsection{Kempe Walks and Variable Eliminations}
In a consistently colored incidence graph, an $(\alpha,\beta)$ Kempe path, or simply $(\alpha,\beta)$ path, where $\alpha,\beta \in C$ and $\alpha \neq \beta$, is a sequence of adjacent links $l_1$, $l_2$, $\ldots$, $l_{n-1}$, $l_n$ such that $\sigma(l_i)\in{\alpha, \beta}$ for $i=1$, $\ldots$, $n$. The vertices contained in the path are called \textbf{\textit{interior}} vertices of the path. There are two types of maximal $(\alpha, \beta)$ paths in a bridgeless cubic graph:
\begin{enumerate}
	\item \textbf{\textit{$(\alpha, \beta)$ cycle}}: The two end-links $l_1$ and $l_n$ are adjacent to each other.
	\item \textbf{\textsl{$(\alpha, \beta)$ open path}}: The two end-links $l_1$  and $l_n$ are not adjacent to each other, and both ends of the path are fictitious vertices.
\end{enumerate}	
An $(\alpha, \beta)$ variable edge is always contained in a maximal $(\alpha, \beta)$ path, either an $(\alpha, \beta)$ cycle or an $(\alpha, \beta)$ open path. Variable eliminations can be achieved by the binary color exchange operation $\otimes$ performed on two adjacent colored edges ${\vec e}_{j,i}=(c_{j,i},c_{i,j})$  and ${\vec e}_{i,k}=(c_{i,k},c_{k,i})$, which is defined as follows:
\begin{equation}
(c_{j,i},c_{i,j})\otimes (c_{i,k},c_{k,i})=(c_{j,i}, \beta)\otimes(\alpha,c_{k,i})\Rightarrow (c_{j,i},\alpha)\circ(\beta,c_{k,i}).
\end{equation}
If the two adjacent edges are variables ${\vec e}_{j,i}=(\alpha, \beta)$ and ${\vec e}_{i,k}=(\alpha, \gamma)$, then the following color exchange operation: 
\begin{equation}
{\vec e}_{j,i}\otimes {\vec e}_{i,k}=(\alpha, \beta)\otimes (\alpha, \gamma)\Rightarrow (\alpha, \alpha)\circ (\beta, \gamma)
\end{equation}
can eliminate one of these variables. In general, variable eliminations require a sequence of color exchanges to move one variable to another variable along a two-colored Kempe path.

The Kempe walk of a $(\alpha,\beta)$ variable on a $(\alpha,\beta)$ path is a sequence of color exchange operations performed on its interior vertices. Examples of variable eliminations by Kempe walks are provided in Fig.~\ref{fig_3}. Consider the $(a,b)$ path $(a,b)\circ (a,a)\circ (b,a)$ shown in Fig.~\ref{fig_3}. The variable ${\vec e}_1=(a,b)$ can walk to another variable ${\vec e}_2=(b,a)$ by the following sequence of color exchanges performed on its interior vertices:
\begin{equation}
(a,b)\otimes (a,a)\circ (b,a) \Rightarrow (a,a)\circ(b,a)\otimes (b,a)\Rightarrow (a,a)\circ (b,b)\circ (a,a),
\end{equation}
in which two variables are eliminated by color exchanges. 

\begin{figure}[htbp]
\centering
\includegraphics[scale=0.8]{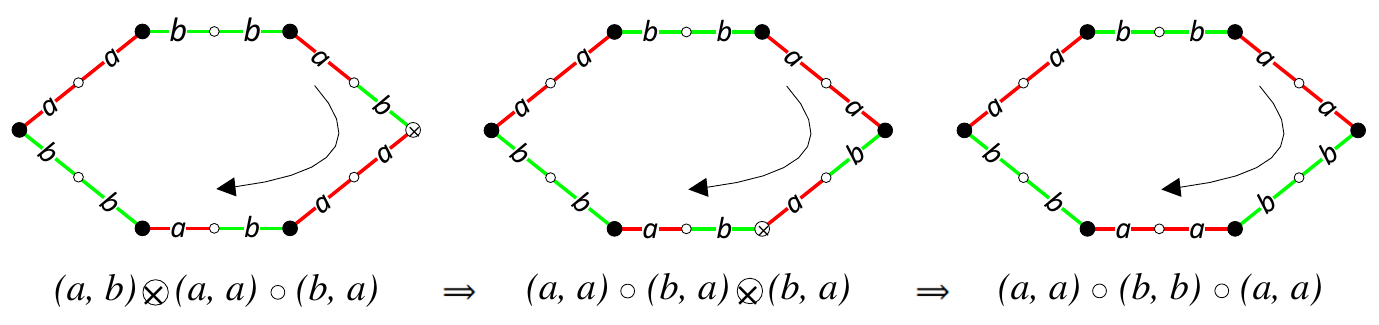}
\caption{Variable eliminations by Kempe walk.}
\label{fig_3}
\end{figure}

Another useful application of the Kempe walk is to negate a variable. The \textbf{\textit{negation}} of a variable $(\alpha, \beta)$, denoted as $-(\alpha, \beta)=(\beta, \alpha)$, represents a color vector in the opposite direction of $(\alpha, \beta)$. The negation requires the color conversion of the entire maximal $(\alpha, \beta)$ path $H$ that contains the variable $(\alpha, \beta)$. That is, the color inversion of the $(\alpha, \beta)$ path $H$ involves a sequence of color exchanges of all interior vertices of $H$. The negation only applies to a maximal $(\alpha, \beta)$ path, either a $(\alpha, \beta)$ cycle or an open $(\alpha, \beta)$ path; otherwise, the operation may introduce new variables.

\section{Decomposition of a Configuration}
\label{sec3}
As the previous section shows, variables can move along two-colored alternate paths via color exchange operations, called \textbf{\textit{Kempe walks}}, and they will cancel each other while moving around the graph. The problem is solved if all variables are eliminated and a proper 3-edge colored cubic graph $G$ is reached; otherwise, $G$ is a Class 2 cubic graph, called a snark. Our study of edge coloring of cubic graphs starts with Petersen’s theorem stated as follows.

\begin{theorem}[Petersen]
Every bridgeless cubic graph contains a perfect matching.
\end{theorem}
This theorem first appeared in \cite{Petersen1898}. Today, it can be proved by an application of the Tutte theorem \cite{chartrand2011graphs, west2001introduction}. In a bridgeless cubic graph $G$ with a perfect matching, the edges that are not in the perfect matching form a set of disjoint cycles, called \textbf{\textit{Tait cycles}}. For any Petersen's perfect matching, assigning color $c$ to the edges in the perfect matching and color $a$ or $b$ to the links in these Tait cycles, we can obtain a complex coloring $T(G)$ of graph $G$ that only contains $(a,b)$ variables. It is easy to show that all $(a,b)$ variables are contained in odd $(a,b)$ Tait cycles, and every odd $(a,b)$ Tait cycle contains exactly one $(a,b)$ variable in such a complex coloring. The complex coloring $T(G)$ corresponding to the Petersen's perfect matching is called a \textbf{\textit{configuration}} of graph $G$ in this paper.

As an example, two configurations of a cubic planar graph $G$ are depicted in Fig.~\ref{fig_4}. The configuration shown in Fig.~\ref{fig_4_a} has two $(a,b)$ variables respectively contained in two disjoint odd $(a,b)$ cycles. The properly colored configuration shown in Fig.~\ref{fig_4_b} contains one even $(a,b)$ cycle, one even $(b,c)$ cycle, and two even $(a,c)$ cycles. 

\begin{figure}[htbp]
 \centering
 \subfigure[Two disjoint odd $(a,b)$ cycles.]{
  \includegraphics[scale=0.8]{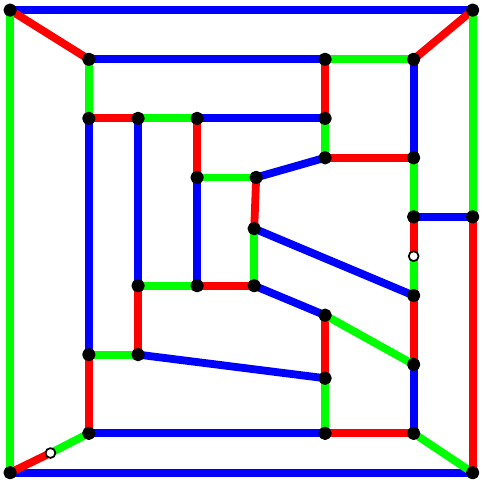}
   \label{fig_4_a}
   } \qquad
 \subfigure[$(a,b)$, $(b,c)$, and $(a,c)$ even cycles.]{
  \includegraphics[scale=0.8]{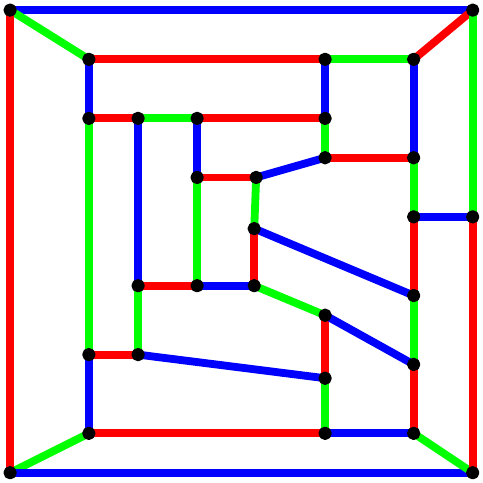}
   \label{fig_4_b}
   }
 \caption{Complex coloring configurations of a cubic planar graph $G$.}  
\label{fig_4}
\end{figure}

For any given Petersen's perfect matching of a bridgeless cubic graph $G$, each $(a,b)$ Tait cycle can be colored in two different ways. One coloring of the $(a,b)$ Tait cycle is the negation of the other, meaning exchange the two $a$, $b$ colors. And each variable ${\vec e}=(a,b)$ can be located at any edge of the odd cycle that contains $\vec e$. Therefore, different color assignments corresponding to the same Petersen's perfect matching form an equivalent class, and they are considered as the same configuration $T(G)$ but in different $a$, $b$ coloring \textbf{\textit{states}}. That is, a configuration $T(G)$ of graph $G$ is uniquely determined by the set of $(c,c)$ edges, a Petersen's perfect matching, while the state of $T(G)$ is determined by $a$ and $b$ links. Let $\tau_o$ and $\tau_e$ be the respective number of odd and even $(a,b)$ Tait cycles in $T(G)$, and $S_{T(G)}$ denote the set of states of $T(G)$. Then we have:
\begin{equation}
|S_{T(G)}|=2^\tau n_1 n_2 \cdots n_{\tau_o},
\end{equation}
where $\tau=\tau_o+\tau_e$ is the total number of $(a, b)$ Tait cycles, and $n_i$ is the length of $i$-th odd $(a,b)$ Tait cycles, for $i=1$, $\ldots$, $\tau_o$, in the configuration $T(G)$. That is, there is a one-to-one correspondence between perfect matchings and configurations of bridgeless cubic graphs.

In a configuration $T(G)$ of a bridgeless cubic graph $G$, any two maximal two-colored sub-graphs $H_1(\alpha,\beta)$ and $H_2(\alpha,\beta)$, for $\alpha,\beta \in C=\{a,b,c\}$, must be vertex disjoint, which implies that any configuration can be decomposed into a set of maximal two-colored $(a,b)$, $(b,c)$ and $(c,a)$ sub-graphs. There are five such kinds of sub-graphs in a configuration, and they are listed as follows:
\begin{itemize}
	\item $(a,b)$ Tait cycles, they can be either odd or even, and each odd Tait cycle contains an $(a,b)$ variable. 
	\item $(a,c)$ and $(b,c)$ even cycles.
	\item $(a,c)$ and $(b,c)$ open paths connecting $(a,b)$ variables. 
\end{itemize}

The collection of these maximal two-colored sub-graphs covers each link of graph $G$ exactly twice, as illustrated by the configuration $T(G)$ of a cubic graph $G$ shown in Fig.~\ref{fig_5_a}. The properties of these maximal two-colored sub-graphs are described as follows:

\begin{enumerate}
	\item \textit{Locking Cycle}
The odd Tait cycle $H(a,b)$ that contains an $(a,b)$ variable is called a \textbf{\textit{locking cycle}}. Two $(a,b)$ locking cycles of the configuration $T(G)$ are shown in Fig.~\ref{fig_5_b}.

\item \textit{Exclusive Chain}
The open $(a,c)$ and $(b,c)$ paths connecting two $(a,b)$ variables are called \textbf{\textit{exclusive chains}}. Negating an $(a,c)$ exclusive chain will change the two end $(a,b)$ variables into two $(b,c)$ variables. Similarly, color inverting of a $(b,c)$ exclusive chain will change the two ends into two $(a,c)$ variables. Therefore, the color inversion of any exclusive chain is prohibited, and the name exclusive chain implies that the two end variables are mutually exclusive by the chain. The $(a,c)$ exclusive chain and the $(b,c)$ exclusive chain of $T(G)$ are shown in Fig.~\ref{fig_5_c} and \ref{fig_5_f}, respectively.

\begin{figure}[tbp]
 \centering
\begin{tabular}{ccc}
	\subfigure[Configuration $T(G)$.]{
		\includegraphics[scale=0.7]{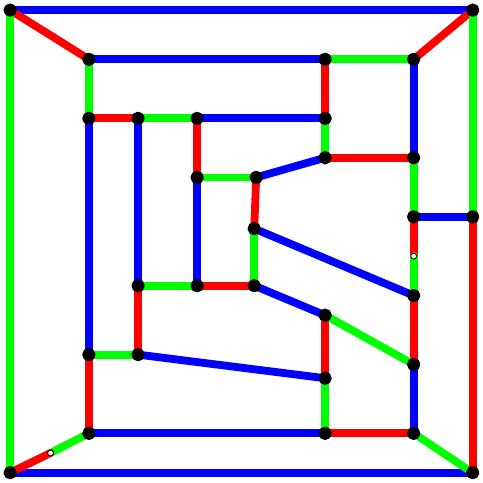}
		\label{fig_5_a}
		}
	&\subfigure[Two locking $(a,b)$ cycles.]{
		\includegraphics[scale=0.7]{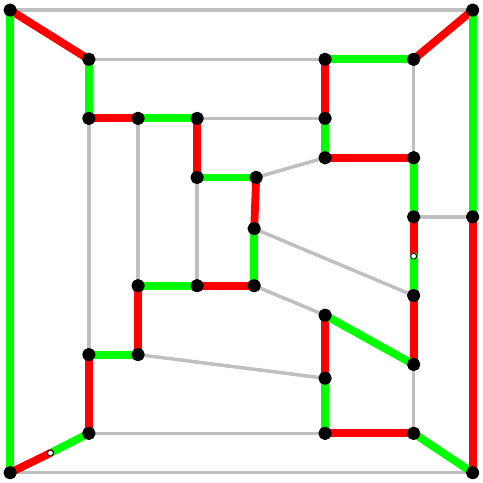}
		\label{fig_5_b}
		}
	 \\
	\subfigure[The $(a,c)$ exclusive chain.]{
		\includegraphics[scale=0.7]{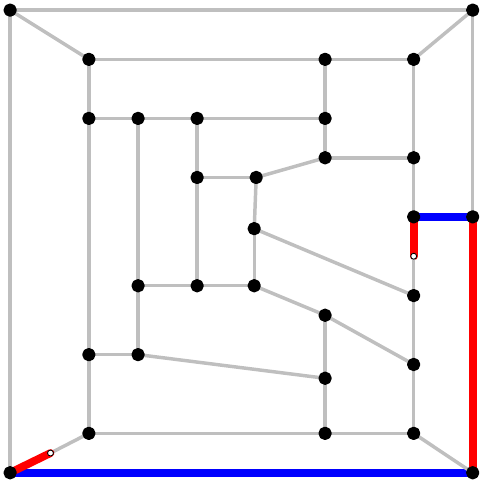}
		\label{fig_5_c}
		} 
	&\subfigure[An essential $(a,c)$ cycle.]{
		\includegraphics[scale=0.7]{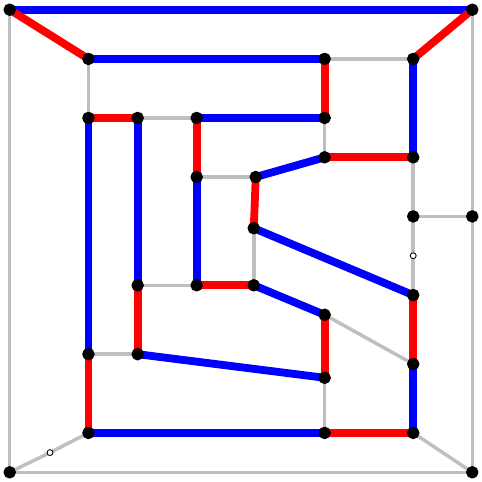}
		\label{fig_5_d}
		} 
	&\subfigure[Two even $(a,b)$ cycles after negating $(a,c)$ cycle.]{
		\includegraphics[scale=0.7]{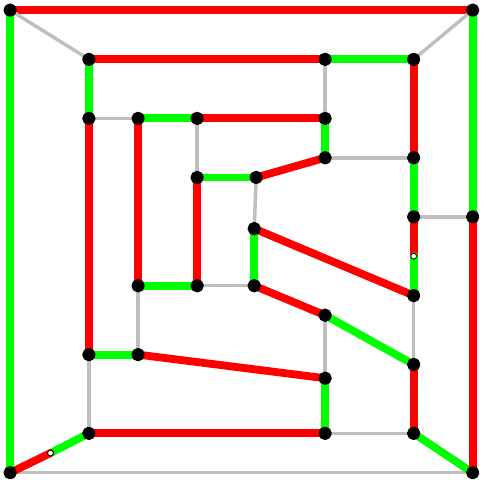}
		\label{fig_5_e}
		} 
	\\
	\subfigure[The $(b,c)$ exclusive chain.]{
		\includegraphics[scale=0.7]{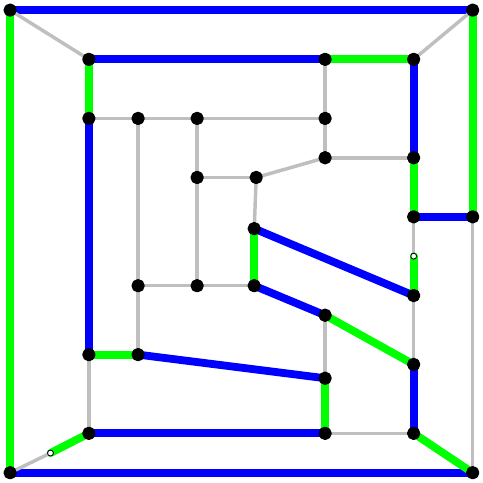}
		\label{fig_5_f}
		}  
	&\subfigure[A nonessential $(b,c)$ cycle.]{
		\includegraphics[scale=0.7]{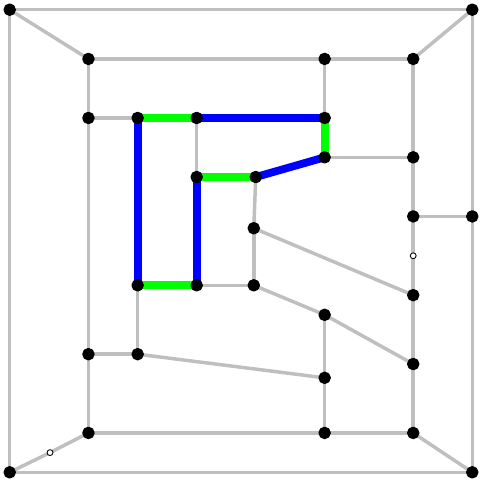}
		\label{fig_5_g}
		} 
	&\subfigure[Two odd $(a,b)$ cycles after negating $(b,c)$ cycle.]{
		\includegraphics[scale=0.7]{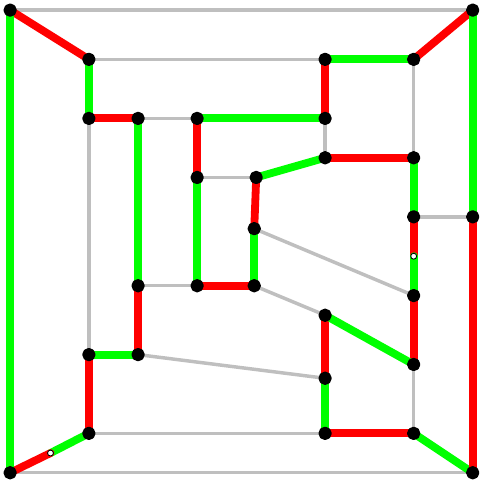}
		\label{fig_5_h}
		}
\end{tabular}
 \caption{Decomposition of a complex coloring configuration $T(G)$.}  
\label{fig_5}
\end{figure}

\item \textit{Resolution Cycle}
A variable contained in a locking cycle $H(a,b)$ can be released by negating another cycle $R_1(a,c)$ which shares some $(a,a)$ edges with $H(a,b)$. After inverting the colors in cycle $R_1(a,c)$, the shared $(a,a)$ edges become $(c, c)$ edges, which deform the original locking cycle $H(a,b)$. Thus, cycle $R_1(a,c)$ is called a \textbf{\textit{resolution cycle}} of the $(a,b)$ variable locked in the locking cycle $H(a,b)$. 

After inverting the colors of the resolution cycle $R_1(a,c)$, if the locked $(a,b)$ variable can be canceled with another $(a,b)$ variable by a walk on a newly created Kempe path, then $R_1(a,c)$ is called an \textbf{\textit{essential cycle}}. Otherwise, it is a \textbf{\textit{nonessential cycle}}. Similarly, a cycle $R_2(b,c)$ that shares some $(b,b)$ edges with $H(a,b)$ also serve as a resolution cycle of the $(a,b)$ variable contained in $H(a,b)$.

\end{enumerate}

The decomposition of a configuration $T(G)$ of graph $G$ is illustrated by the running example shown in Fig.~\ref{fig_5}, from which we observe the following additional properties of these maximal two-colored sub-graphs of $T(G)$:
\begin{enumerate}
	\item The $(a,c)$ exclusive chain and the $(a,c)$ resolution cycle are disjoint, as shown in Fig.~\ref{fig_5_c} and \ref{fig_5_d}, respectively, but their union consists of all $a$ and $c$ links of $T(G)$. 
	\item A similar property holds for $b$ and $c$ links of $T(G)$, as shown in Fig.~\ref{fig_5_f} and \ref{fig_5_g}. 
	\item The collection of all links in the two locking $(a,b)$ cycles, the $(a,c)$ exclusive chain, the $(b,c)$ exclusive chain, the $(a,c)$ resolution cycle, and the $(b,c)$ resolution cycle includes every link in $T(G)$ exactly twice.
\end{enumerate}

Next, we consider the negation of the two resolution cycles, one $(a,c)$ cycle and one $(b,c)$ cycle, of the configuration $T(G)$, shown in Fig.~\ref{fig_5_d} and \ref{fig_5_g}, respectively. If we negate the $(a,c)$ cycle, the resulting configuration has two even $(a,b)$ cycles, as shown in Fig.~\ref{fig_5_e}, and the two $(a,b)$ variables contained in the same $(a,b)$ cycle can be easily eliminated by a Kempe walk. Thus, the $(a,c)$ resolution cycle displayed in Fig.~\ref{fig_5_d} is an essential cycle.

On the other hand, the coloring resulting from the negation of the $(b,c)$ cycle still has two disjoint odd $(a,b)$ cycles, each of which contains an $(a,b)$ variable, as shown in Fig.~\ref{fig_5_h}. Therefore, the $(b,c)$ resolution cycle displayed in Fig.~\ref{fig_5_g} is a nonessential cycle.

\section{Reducibility of a Configuration}
\label{sec4}
A configuration $T(G)$ of a bridgeless cubic graph $G$ is uniquely determined by the $(c,c)$ edges in Petersen's perfect matching, but the states of $T(G)$ are determined by $a$ and $b$ links. Thus, the following two operations will cause state transitions of $T(G)$ but not change the configuration itself: 
	\begin{itemize}
		\item Negate any $(a,b)$ cycle, either even or odd.
		\item Move any $(a,b)$ variable within its locking cycle.
	\end{itemize}
Notice that the state transitions due to the above two operations will retain the sub-graphs of all $(a,b)$ cycles intact; however, they will change $(a,c)$ and $(b,c)$ exclusive chains and resolution cycles.

Let $S_{T(G)}$ denote the set of all states of $T(G)$. We say that a state $\xi \in S_{T(G)}$ is \textbf{\textit{reducible}} if one of the even $(a,c)$ or $(b,c)$ cycles in the state $\xi$ is essential; otherwise, the state $\xi \in S_{T(G)}$ is \textbf{\textit{irreducible}}. If the state $\xi$ is irreducible, then we can implement the above state transition operations to change the state of $T(G)$ to another state. Since the number of states is finite, given by $|S_{T(G)}|=2^\tau n_1 n_2 \ldots n_{\tau_o}$, it is feasible to verify the reducibility of all states of $T(G)$ in a systematic manner by a deterministic algorithm.

A configuration $T(G)$ is \textbf{\textit{reducible}} if one of the states $\xi \in S_{T(G)}$ is reducible, meaning that two $(a,b)$ variables in $T(G)$ can be eliminated by a Kempe walk after negating an essential cycle in the state $\xi$. The resulting configuration $T'(G)$ is a proper coloring of graph $G$ if it no longer contains any remaining $(a,b)$ variables.

By contrast, a configuration $T(G)$ is \textbf{\textit{irreducible}} if all states are irreducible. An irreducible configuration $T(G)$ does not imply that the underlying cubic graph is Class 2, because configuration $T(G)$ can be transformed into another configuration $T'(G)$ by negating some even $(a,c)$ or $(b,c)$ cycles in $T(G)$. 

By collecting the above discussions, as illustrated by the transition diagram shown in Fig.~\ref{fig_6}, we can repeatedly use the following two operations to find a reducible configuration:

\begin{figure}[htbp]
  \centering
  \subfigure[Transitions of configurations.]{
     \includegraphics[scale=0.9]{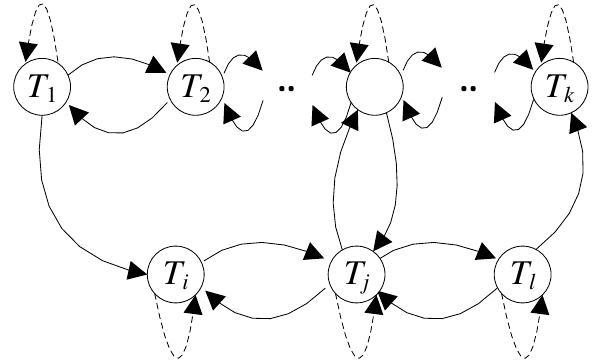}
     \label{fig_6_a}
     } 
  \subfigure[Closed set of irreducible configurations.]{
     \includegraphics[scale=0.9]{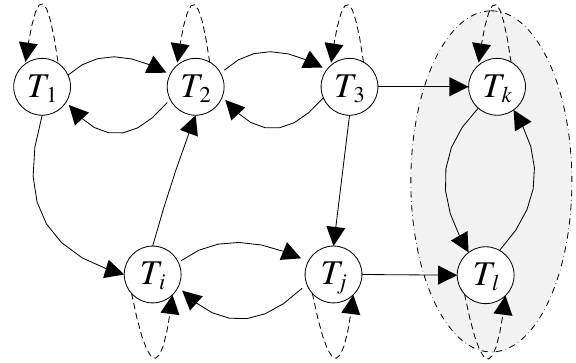}
     \label{fig_6_b}
     }
	\subfigure{\includegraphics[scale=0.9]{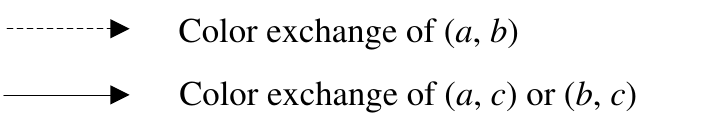}}
  \caption{The transition diagram of configurations.}
	\label{fig_6}
\end{figure}

\begin{enumerate}
	\item \textbf{\textit{Local operation}} - move to another state within the same configuration $T(G)$:
		
Local operations only involve $(a,b)$ color exchanges, which leave all $(c,c)$ edges in the Petersen perfect matching intact. We can systematically move $(a,b)$ variables within their $(a,b)$ locking cycles until the configuration reaches a reducible state. The maximum number of moves within the state space $S_{T(G)}$ is bounded by $|S_{T(G)} |=2^\tau n_1 n_2 \ldots n_{\tau_o}$.

	\item \textbf{\textit{Global operation}} - transform $T(G)$ into another configuration $T'(G)$:
		
Global operations involve $(a,c)$ and $(b,c)$ color exchanges, which will change the Petersen perfect matching that determines the configuration. If configuration $T(G)$ is reducible, then it will transform into another configuration $T'(G)$ after eliminating two $(a,b)$ variables. On the other hand, if configuration $T(G)$ is irreducible, then we can transform $T(G)$ into another configuration $T'(G)$ by negating some even $(a,c)$ or $(b,c)$ cycles in $T(G)$. The process of searching a reducible configuration is a random walk in the entire space of configurations of the cubic graph $G$. 
\end{enumerate}

\begin{figure}[tbp]
 \centering
 \subfigure[An irreducible configuration $T(G)$.]{
  \includegraphics[scale=0.7]{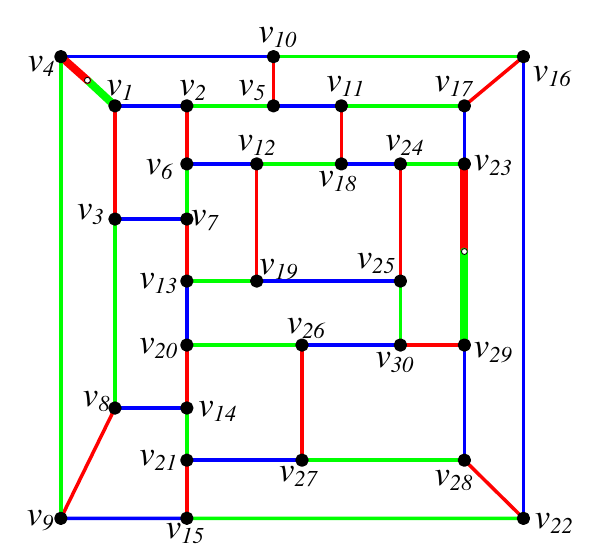}
   \label{fig_7_a}
   } \qquad
 \subfigure[Move an $(a,b)$ variable to edge $(v_1,v_3)$.]{
  \includegraphics[scale=0.7]{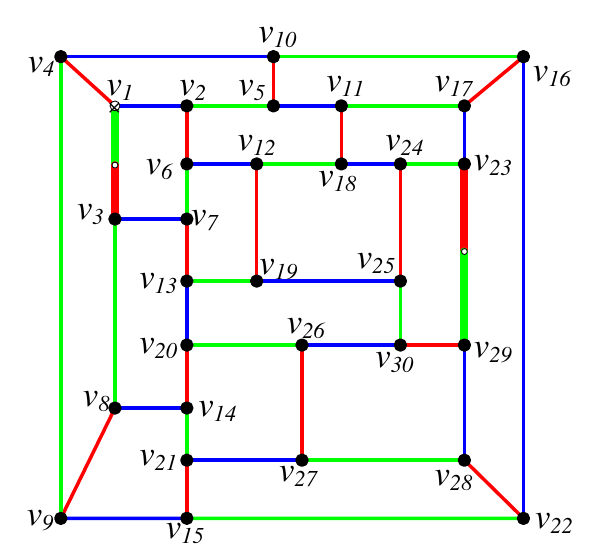}
   \label{fig_7_b}
   } \qquad
	\subfigure[Negate the highlighted $(a,c)$ cycle.]{
  \includegraphics[scale=0.7]{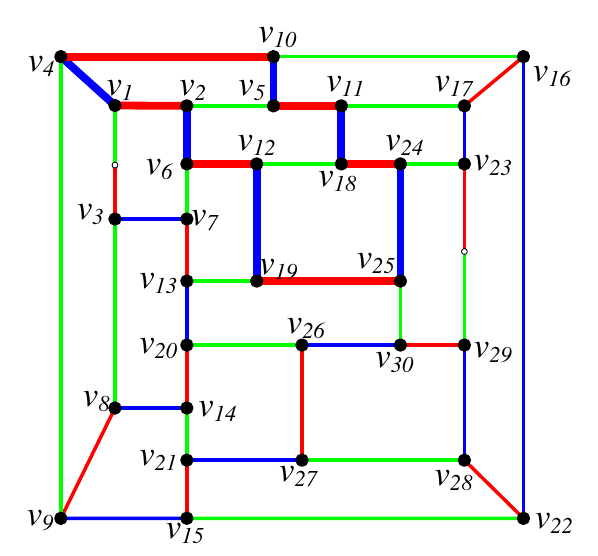}
   \label{fig_7_c}
   } \qquad
 \subfigure[Negate the highlighted $(b,c)$ cycle.]{
  \includegraphics[scale=0.7]{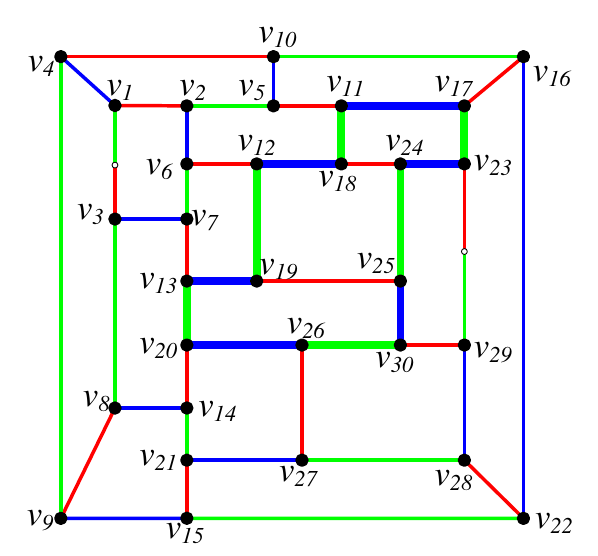}
   \label{fig_7_d}
   } \qquad
	\subfigure[Two $(a,b)$ variables are connected by a Kempe path.]{
  \includegraphics[scale=0.7]{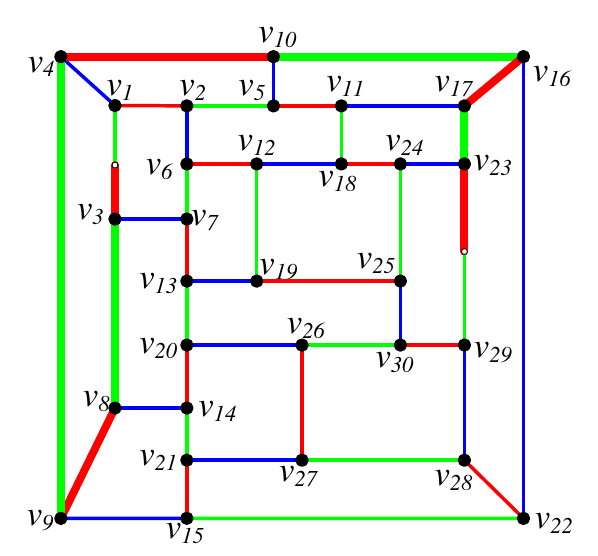}
   \label{fig_7_e}
   }
 \caption{Transforming an irreducible configuration into a reducible configuration.}  \label{fig_7}
\end{figure}

An irreducible configuration $T(G)$ may transform into a reducible configuration $T'(G)$ by a sequence of negating $(a,c)$ or $(b,c)$ cycles. For example, the following sequence of operations will eliminate the two $(a,b)$ variables in the irreducible configuration $T(G)$ shown in Fig.~\ref{fig_7_a}:

\begin{itemize}
		\item Move one of the $(a,b)$ variables to the edge $(v_1,v_3)$, as shown in Fig.~\ref{fig_7_b}.
		\item Negate the $(a,c)$ cycle $(v_1-v_2-v_6-v_{12}-v_{19}-v_{25}-v_{24}-v_{18}-v_{11}-v_5-v_{10}-v_4-v_1)$, as highlighted in Fig.~\ref{fig_7_c}. 
		\item Negate the $(b,c)$ cycle $(v_{11}-v_{18}-v_{12}-v_{19}-v_{13}-v_{20}-v_{26}-v_{30}-v_{25}-v_{24}-v_{23}-v_{17}-v_{11})$, as highlighted in Fig.~\ref{fig_7_d}. 
		\item The two variables are now connected by the $(a,b)$ path $(v_{23}-v_{17}-v_{16}-v_{10}-v_4-v_9-v_8-v_3)$, as shown in Fig.~\ref{fig_7_e}, and they can be eliminated by a Kempe walk. 
\end{itemize}

The bridgeless cubic graph $G$ that cannot have a proper 3-edge coloring is commonly referred to as a snark \cite{Belcastro2012, Holton1993}. The best known snark is the Petersen graph $G_P$, which is the smallest bridgeless cubic graph with no 3-edge coloring, as shown in Fig.~\ref{fig_8_a}. It can be easily shown that the configurations of the Petersen graph are all equivalent by relabeling vertices, meaning that they are isomorphic to each other. This unique configuration $T(G_P)$ of the Petersen graph is irreducible. Although $T(G_P)$ has $4\times 5\times 5=100$ states, but they are all isomorphic to one of the two states shown in Fig.~\ref{fig_8_b} and \ref{fig_8_e}, respectively, and neither state is reducible.

One of the irreducible states $\xi_1$ of $T(G_P)$ is shown in Fig.~\ref{fig_8_b}, in which the two $(a,b)$ variables are respectively contained in two disjoint $(a,b)$ cycles, and there are no $(a,c)$ or $(b,c)$ resolution cycles. The two $(a,c)$ and $(b,c)$ exclusive chains in state $\xi_1$ are displayed in Fig.~\ref{fig_8_c} and \ref{fig_8_d}, respectively.

Similarly, another irreducible state $\xi_2$ of $T(G_P)$ is shown in Fig.~\ref{fig_8_e}. Besides the two $(a,c)$ and $(b,c)$ exclusive chains, as displayed in Fig.~\ref{fig_8_f} and \ref{fig_8_g}, respectively, the state $\xi_2$ of $T(G_P)$ has one $(a,c)$ nonessential resolution cycle. 

\begin{figure}[tbp]
 \centering
		 \begin{tabular}{ccc}
		\ 
		& \subfigure[The Petersen graph.]{
				 \includegraphics[scale=0.8]{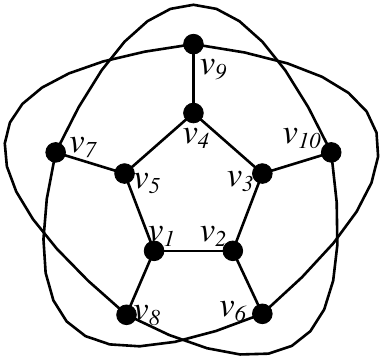}
				 \label{fig_8_a}
				 }
				&\ 
				 \\
			\subfigure[A state $\xi_1$ of the Petersen graph.]{
				\includegraphics[scale=0.8]{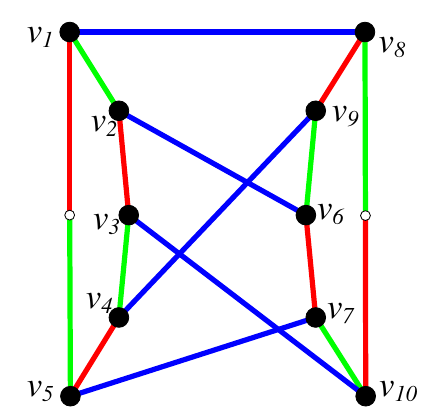}
				\label{fig_8_b}
				} 
				&\subfigure[$(a,c)$ sub-graph of $\xi_1$.]{
				\includegraphics[scale=0.8]{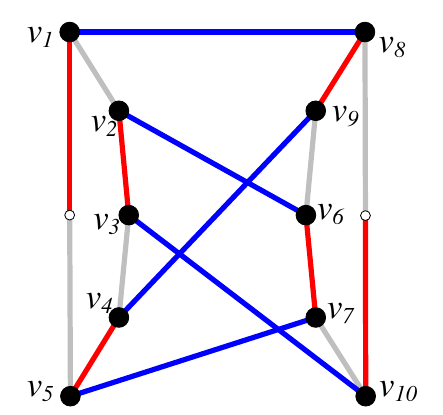}
				\label{fig_8_c}
				} 
			&\subfigure[$(b,c)$ sub-graph of $\xi_1$.]{
				\includegraphics[scale=0.8]{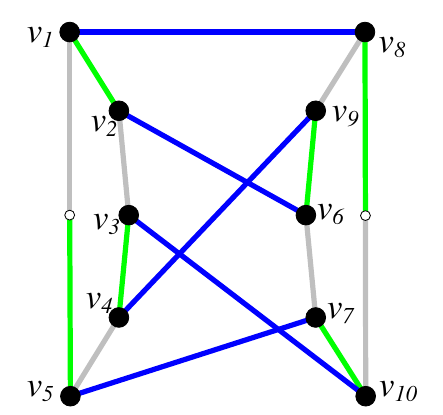}
				\label{fig_8_d}
				} \\
				\subfigure[A state $\xi_2$ of the Petersen graph.]{
				\includegraphics[scale=0.8]{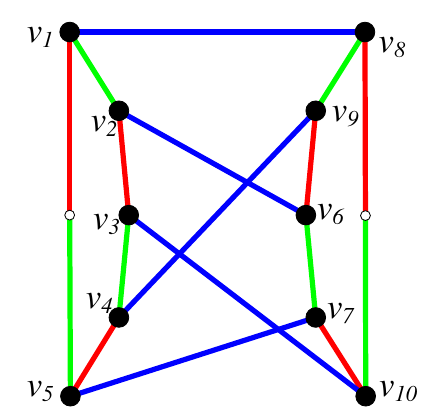}
				\label{fig_8_e}
				} 
			&\subfigure[$(a,c)$ sub-graph of $\xi_2$.]{
				\includegraphics[scale=0.8]{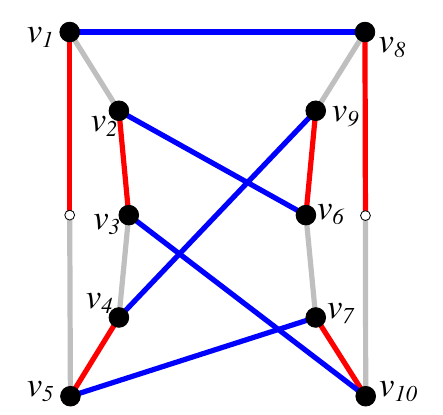}
				\label{fig_8_f}
				} 
				&\subfigure[$(b,c)$ sub-graph of $\xi_2$.]{
				\includegraphics[scale=0.8]{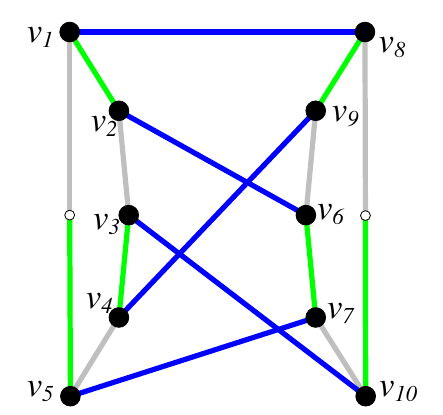}
				\label{fig_8_g}
				}
		\end{tabular}

  \caption{Two irreducible states of a configuration of the Petersen graph.}
  \label{fig_8}
\end{figure}

It is easy to show that not only are all configurations of the Petersen graph irreducible; in fact, they are isomorphic to each other. As an example, negating the $(a,c)$ cycle $(v_2-v_6-v_9-v_4-v_5-v_7-v_{10}-v_3-v_2)$ in the configuration $T(G_p)$ shown in Fig.~\ref{fig_9_a}, we obtain another configuration $T'(G_p)$ displayed in Fig.~\ref{fig_9_b}, which is the same as $T(G_p)$ by relabeling corresponding vertices of the Petersen graph. Obviously, these configurations are all irreducible because the two $(a,b)$ variables can never be connected by an $(a,b)$ Kempe path.

\begin{figure}[htbp]
 \centering
 \subfigure[A configuration $T(G_p)$ of the Petersen graph.]{
  \includegraphics[scale=0.8]{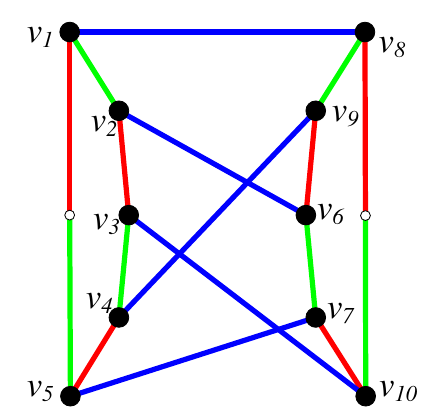}
   \label{fig_9_a}
   } \quad
 \subfigure[Negate $(a, c)$ resolution cycle.]{
  \includegraphics[scale=0.8]{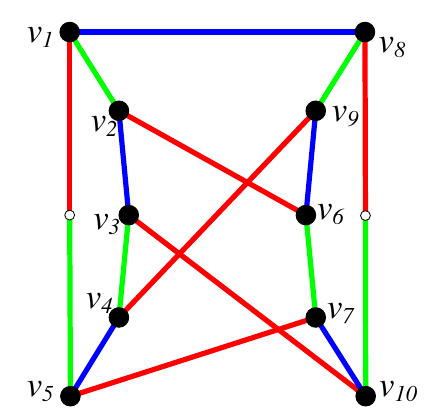}
   \label{fig_9_b}
   } \quad
\subfigure[Re-arrange to another configuration $T'(G_p)$.]{
  \includegraphics[scale=0.8]{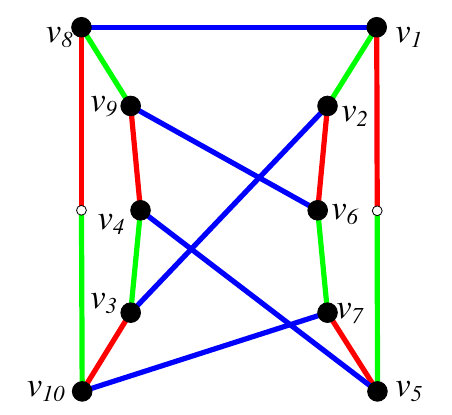}
   \label{fig_9_c}
   } \quad
 \caption{Isomorphic configuration of the Petersen graph.} 
\label{fig_9}
\end{figure}

In 1966, Tutte conjectured that every bridgeless cubic graph that does not contain the Petersen graph as a graph minor is 3-edge colorable, or equivalently, every snark has the Petersen graph as a graph minor \cite{Tutte196615}. Neil Robertson and Robin Thomas announced in 1996 that they proved this conjecture \cite{Belcastro2012, Robertson1997166, Thomas99recentexcluded}, but they did not publish the result. If this conjecture is valid, then the 4CT can be immediately established according to Tait's equivalent formulation.

Tutte's conjecture holds for almost all known snarks. The contraction processes of some well-known snarks are illustrated in Appendix \ref{appdx_A}. The Petersen graph as a graph minor is not a proper characterization of snarks. It is easy to show that many 3-edge colorable graphs also have the Petersen graph as a graph minor. For example, the 3-edge colorable cubic graph shown in Fig.~\ref{fig_10} is obtained by adding the edge $e_{11,12}=(v_{11},v_{12})$ to the Petersen graph, which certainly has the Petersen graph as a graph minor.

\begin{figure}[htbp]
 \centering
 \subfigure[Add an edge to the Petersen graph.]{
  \includegraphics[scale=0.7]{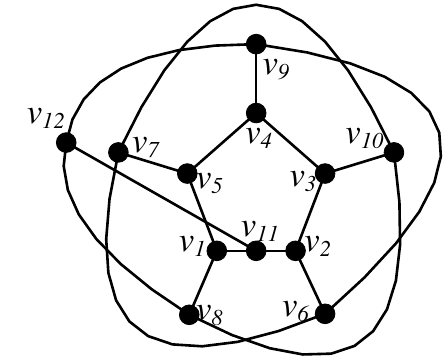}
   \label{fig_10_a}
   } 
 \qquad
 \subfigure[A 3-edge coloring of the modified Petersen graph.]{
  \includegraphics[scale=0.7]{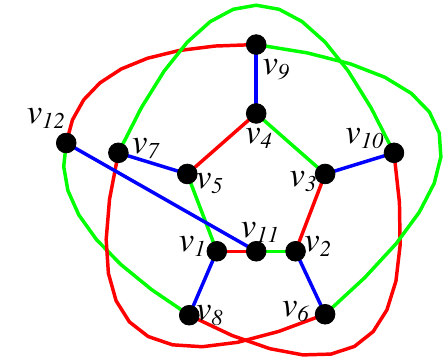}
   \label{fig_10_b}
   }
 \caption{A 3-edge colorable cubic graph with the Petersen graph as graph minor.} 
\label{fig_10}
\end{figure}

Since increasing the number of variables is prohibited in any transformation of a configuration, it is possible that a set of irreducible configurations is closed in a transition diagram under the transformations defined above, as illustrate in Fig. ~\ref{fig_6_b}, in which an irreducible configuration can only transform into other irreducible configurations. Therefore, an immediate consequence is the characterization of snarks given in the following theorem.

\begin{theorem}[]
A bridgeless cubic graph $G(V,E)$ is a Class 2 graph if and only if $G$ has a closed set of irreducible configurations.
\end{theorem}

Note that, for the same reason stated above, these irreducible configurations in the closed set of the snark should all possess the same minimum number of variables.

\section{Petersen Configuration}
\label{sec5}
For edge coloring of bridgeless cubic planar graphs, we are interested in a particular configuration $P(G)$, referred to as the \textbf{\textit{Petersen configuration}}, which satisfies the following conditions:
	
	\begin{enumerate}
		\item The configuration $P(G)$ contains two $(a,b)$ variables.
		\item The two $(a,b)$ variables are on the boundary of a pentagon in some state $\xi$ of $P(G)$.
	\end{enumerate}	
It is easy to show that the above conditions imply that three edges of the pentagon are contained in the two odd $(a,b)$ Tait cycles, and the remaining two $(c,c)$ edges of the pentagon belong to the perfect matching. The state $\xi$ of $P(G)$ shown in Fig.~\ref{fig_11} satisfies both conditions.

\begin{figure}[htbp]
\centering
\includegraphics[scale=0.7]{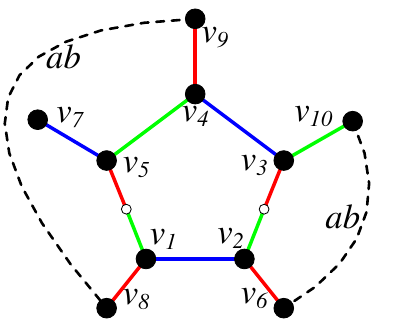}
\caption{A particular state $\xi$ of a Petersen configuration $P(G)$.}
\label{fig_11}
\end{figure}

The two $(a,b)$ variables in a Petersen configuration $P(G)$ can be transformed into two $(a,c)$ variables by color exchanges performed at the two vertices $v_1$ and $v_2$, as shown in Fig.~\ref{fig_12_a}. The result is displayed in Fig.~\ref{fig_12_b}, in which the two $(a,c)$ variables are contained in two disjoint $(a,c)$ cycles; otherwise, they can be easily canceled by a Kempe walk. Finally, the complete state $\xi$ of a Petersen configuration $P(G)$ under consideration is depicted in Fig.~\ref{fig_12_c}.

\begin{figure}[htbp]
 \centering
 \subfigure[Two disjoint $(a,b)$ cycles.]{
  \includegraphics[scale=0.7]{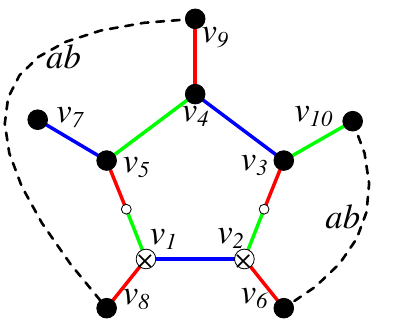}
   \label{fig_12_a}
   } \qquad
 \subfigure[Two disjoint $(a,c)$ cycles.]{
  \includegraphics[scale=0.7]{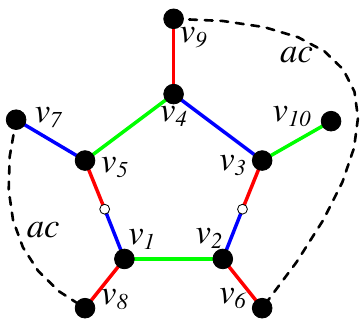}
   \label{fig_12_b}
   } \qquad
	\subfigure[The complete state $\xi$.]{
  \includegraphics[scale=0.7]{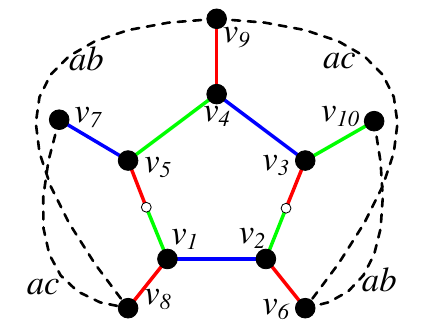}
   \label{fig_12_c}
   }
 \caption{The complete state of a Petersen configuration $P(G)$.}  
\label{fig_12}
\end{figure}

If we perform $(a,b)$ color exchange operations at vertices $v_2$, $v_4$ and $v_5$, as shown in Fig.~\ref{fig_13_a}, the result is displayed in Fig.~\ref{fig_13_b}, in which the right $(a,c)$ chain becomes an $(a,c)$ exclusive chain that connects the two $(a,b)$ variables, and the left $(a,c)$ chain becomes a part of the $(a,c)$ resolution cycle $(v_7-v_5-v_4-v_3-v_2-v_1-v_8-\cdots-v_7)$ that includes all vertices, $v_1$, $v_2$, $v_3$, $v_4$, and $v_5$, of the pentagon. If this $(a,c)$ resolution cycle is essential, then the two $(a,b)$ variables can be canceled, and the Petersen configuration $P(G)$ is reducible. On the other hand, if this $(a,c)$ resolution cycle is nonessential, then we can test other states of this Petersen configuration by moving $(a,b)$ variables or negating $(a,b)$ cycles.

\begin{figure}[htbp]
\centering
	\subfigure[$(a,b)$ color exchanges.]{
		\includegraphics[scale=0.7]{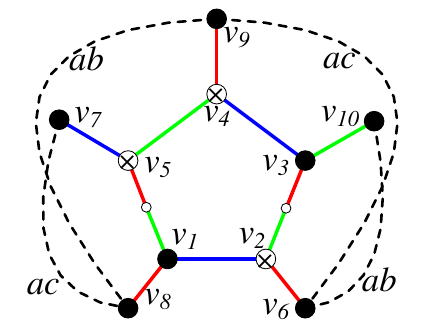}
		\label{fig_13_a}
		} \qquad
 \subfigure[$(a,c)$ resolution cycle and $(a,c)$ exclusive chain.]{
		\includegraphics[scale=0.7]{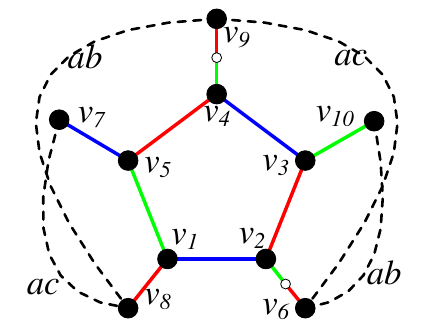}
		\label{fig_13_b}
		} \qquad
	\subfigure[The state $\xi$ of the Petersen graph.]{
		\includegraphics[scale=0.7]{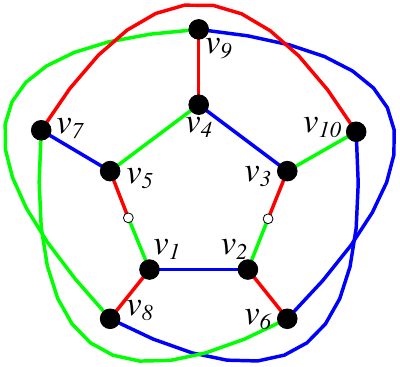}
		\label{fig_13_c}
		}
\caption{The Tait cycle structure of Petersen configuration $P(G)$.}
\label{fig_13}
\end{figure}

A Petersen configuration $P(G)$ is irreducible if all states of $P(G)$ are irreducible. The configuration of the Petersen graph shown in Fig.~\ref{fig_13_c} is the smallest irreducible Petersen configuration $P(G)$. For planar Petersen configurations, we have experimentally tested over one hundred thousand instances generated by computer, and found that they are all reducible. Since any two-colored paths of the same kind cannot cross each other in the plane, it is evident by our perception that an irreducible Petersen configuration $P(G)$ with the Tait cycle structure shown in Fig.~\ref{fig_13_b} must be non-planar. Therefore, we propose the following postulate: 

\begin{repos*}
Every Petersen configuration $P(G)$ of a bridgeless cubic planar graph $G(V,E)$ is reducible.
\end{repos*}

Despite the fact that this proposition has been verified by more than one hundred thousand instances, we still don't have a logical proof of this assertion. We observed the following properties of a Petersen configuration from experimental results:
\begin{enumerate}
\item The cardinality of the state space $S_{P(G)}$ of a Petersen configuration $P(G)$ is given by $|S_{P(G)}|=4 \times 2^{\tau_e} n_1 n_2$, where $\tau_e$ is the number of even $(a,b)$ cycles, and $n_1$ and $n_2$ are the respective numbers of vertices in the two odd $(a,b)$ cycles. In experimental testing of the reducibility of the Petersen configuration, we did not consider the negation of even $(a,b)$ cycles to simplify the computational complexity. Yet we were still able to find reducible states in the reduced state space of size $4 n_1 n_2$. However, in theory, it is not clear whether we can always ignore the negation of even $(a,b)$ cycles.

\item For each Petersen configuration $P(G)$, there is a companion configuration $P'(G)$ that contains two $(a,c)$ variables, as shown in Fig.~\ref{fig_12_b}. Thus, the configuration $P(G)$ can be considered as reducible if any state of $P'(G)$ is reducible. In our experimental testing, however, we have never encountered any case in which the configuration $P(G)$ is irreducible but the companion configuration $P'(G)$ is reducible. 
\end{enumerate}

Furthermore, our experimental results show that there are usually many reducible states in a Petersen configuration $P(G)$ of a bridgeless cubic planar graph $G$. As an example, a reducible $P(G)$ of the planar graph $G$ is shown in Fig.~\ref{fig_14}. The pentagon $(v_{13}-v_7-v_6-v_{12}-v_{19}-v_{13})$ of graph $G$ is highlighted in Fig.~\ref{fig_14_a}, in which the two $(a,b)$ variables are located at edge $(v_7,v_{13})$ and $(v_{12},v_{19})$, respectively. The two locking $(a,b)$ cycles with length 9 and 21, respectively, are displayed in Fig.~\ref{fig_14_b}. This configuration $P(G)$ has $4\times 9\times 21=756$ states, among which 284 of them are reducible. 

\begin{figure}[htbp]
 \centering
 \subfigure[A Petersen configuration.]{
  \includegraphics[scale=0.8]{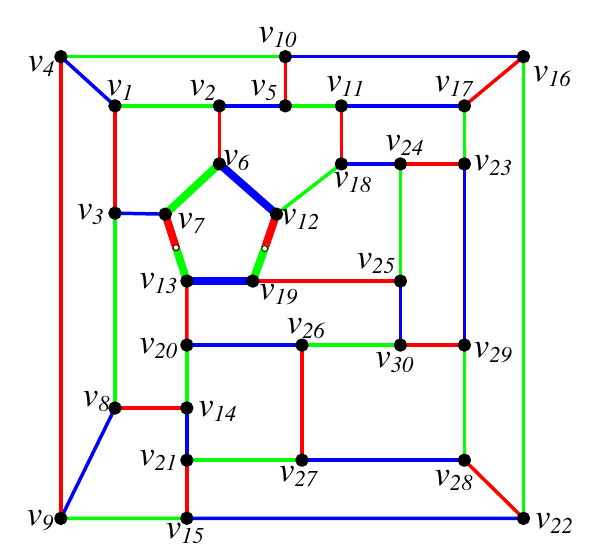}
   \label{fig_14_a}
   } \qquad
 \subfigure[Two highlighted $(a,b)$ locking cycles.]{
  \includegraphics[scale=0.8]{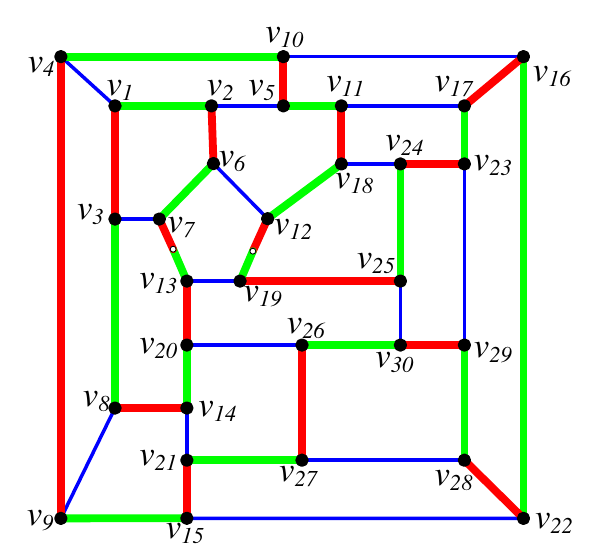}
   \label{fig_14_b}
   }
 \caption{A Petersen configuration $P(G)$ with 284 reducible states.} 
\label{fig_14}
\end{figure}

It should be noted that the second condition of a Petersen configuration $P(G)$ is necessary in the postulate of reducibility, which requires that both odd $(a,b)$ cycles of $P(G)$ contain the boundary edges of the same pentagon. As an example, the configuration $T(G)$ shown in Fig.~\ref{fig_7_a} has two odd $(a,b)$ cycles, each of length 5, and two even $(a,b)$ cycles. The two odd $(a,b)$ cycles contain boundary edges of two different pentagons; therefore $T(G)$ does not satisfy the second condition of a  Petersen configuration. We have checked that the $2^4 \times 5\times 5=400$ states of this configuration $T(G)$ are all irreducible.

\section{Three-edge Coloring Theorem}
\label{sec6}
In this section, we provide an algorithmic approach to prove that every bridgeless cubic planar graph $G(V,E)$ is 3-edge colorable. The basic idea is to recursively color the cubic graph by induction on the number of vertices $|V|$. Initially, it is trivial to show that the smallest simple cubic graph $G$ with $|V|=4$ and $|E|=6$ is 3-edge colorable. For a bridgeless cubic planar graph $G$ with $|V|=n$, we prove that $G$ is 3-edge colorable if the cubic graph $G'$ obtained by deleting an edge in $G$ is 3-edge colorable. The induction steps involve two operations, \textbf{\textit{edge deletion}} and \textbf{\textit{edge insertion}}, defined as follows.

\begin{itemize}
	\item An edge-deletion operation is performed on an uncolored cubic graph $G$. As shown in Fig.~\ref{fig_15_a}, a new cubic graph $G'$ with $|V|-2$ vertices and $|E|-3$ edges can be obtained from $G(V,E)$ by deleting an edge $e$ and smoothing out the two end nodes with degree 2.
	\item An edge-insertion operation shown in Fig.~\ref{fig_15_b} is performed on a colored configuration $T(G)$ of the cubic graph $G$. The insertion operation will introduce two new variables at the two ends of the inserted edge. 
\end{itemize}

\begin{figure}[htbp]
 \centering
 \subfigure[Illustration of edge deletion.]{
  \includegraphics[scale=0.8]{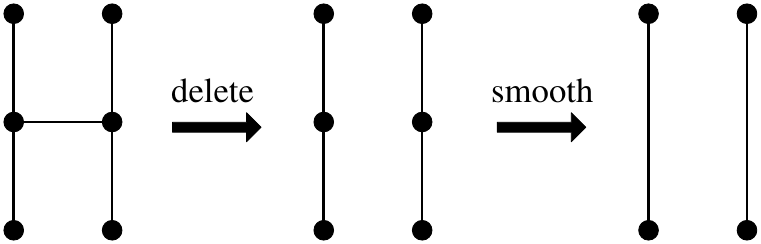}
   \label{fig_15_a}
   } \\	
 \subfigure[Illustration of edge insertion.]{
  \includegraphics[scale=0.8]{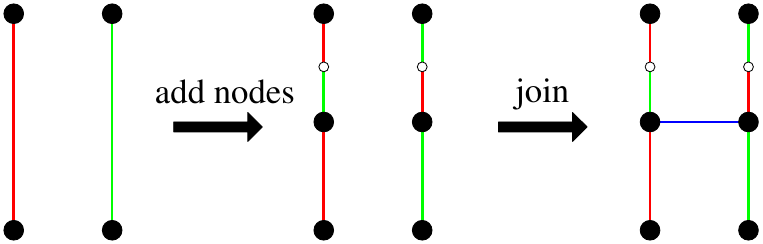}
   \label{fig_15_b}
   }
 \caption{Illustrations of edge deletion and edge insertion.} 
\label{fig_15}
\end{figure}

In graph theory, the \textbf{\textit{girth}} of a graph is the length of the shortest cycle contained in the graph. The girth of a planar graph $G$ is the minimum number of edges surrounding a face in $G$. We need the following two lemmas in the inductive steps of our algorithm.

\begin{lemma}
\label{lem1}
The girth of a bridgeless cubic planar graph $G(V,E)$ is less than or equal to 5.
\end{lemma}

\begin{proof}
Suppose the cubic graph $G(V,E)$ embedded on a sphere has $f$ faces, $v$ vertices, and $e$ edges. From Euler's formula, we have 
\begin{equation}\label{equ7}
v-e+f=2.
\end{equation}
Suppose the girth of $G$ is larger than 5, then we have 
\begin{equation}\label{equ8}
6f \leq 2e.
\end{equation}
In a cubic graph $G(V,E)$, we know that $3v=2e$. From \eqref{equ7}, the above inequality \eqref{equ8} implies $6f \leq 6f-12$, which is impossible.   
\end{proof}

An edge $e$ in a bridgeless cubic planar graph $G(V,E)$ is an \textbf{\textit{admissible edge}}, if the graph $G'$ obtained by deleting edge $e$ remains bridgeless. The following lemma obviously holds in any bridgeless cubic planar graph.
\begin{lemma}
\label{lem2}
Any face of a bridgeless cubic planar graph $G(V,E)$ has at least one admissible edge.
\end{lemma}

In the rest of this section, we show that the following 3-edge coloring theorem is an immediate consequence of the postulate of reducibility of the Petersen configuration. 

\begin{theorem}
Every bridgeless cubic planar graph $G(V,E)$ has a 3-edge coloring.
\end{theorem}
\begin{proof}
As we mentioned above, the 3-edge coloring of a bridgeless cubic planar graph $G(V,E)$ with $|V|=n$ can be derived from the coloring of the cubic graph $G'$, which is obtained by deleting an edge in $G$. Thus, our induction step starts with the selection of a face $F$ in graph $G$ with the minimum number of boundary edges, which is less than 6 according to Lemma~\ref{lem1}. The face $F$ has an admissible edge $e$ according to Lemma~\ref{lem2}, such that the cubic graph $G'$, derived from $G$ by deleting edge $e$, is still bridgeless. 

Suppose that any bridgeless cubic planar graphs with $|V|=n-2$ vertices are 3-edge colorable. We prove that graph $G$ is 3-edge colorable by considering all possible girths of graph $G$. First, it is trivial to show that the induction is valid if the girth is equal to 2. The remaining three nontrivial cases are described as follows. 

\begin{enumerate}
	\item The girth of $G$ equals 3.
	
  Suppose the cubic graph $G(V,E)$ with $n$ vertices has a face $F(v_{2},v_3,v_4)$ with three boundary edges as shown in Fig.~\ref{fig_16_a}. We assume, without loss of generality, that edge $e_{3,4}=(v_3,v_4)$ is admissible. Then a bridgeless cubic graph $G'$ with $n-2$ vertices can be obtained by deleting the edge $e_{3,4}$ and smoothing out the two end vertices $v_3$ and $v_4$. By induction hypothesis, the cubic graph $G'$ has a 3-edge coloring. Suppose the two edges $e_{2,5}=(v_2,v_5)$ and $e_{2,6}=(v_2,v_6)$ in $G'$ are colored by $a$ and $b$, respectively, as shown in Fig.~\ref{fig_16_b}. Then we can derive a 3-edge coloring of graph $G$ as shown in Fig.~\ref{fig_16_d}.
	
\begin{figure}[htbp]
 \centering
 \subfigure[A triangle in $G$.]{
  \includegraphics[scale=1]{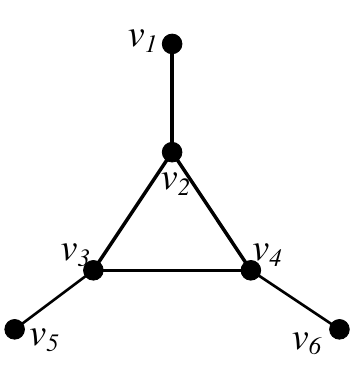}
   \label{fig_16_a}
   } \qquad
 \subfigure[A 3-edge coloring of $G'$.]{
  \includegraphics[scale=1]{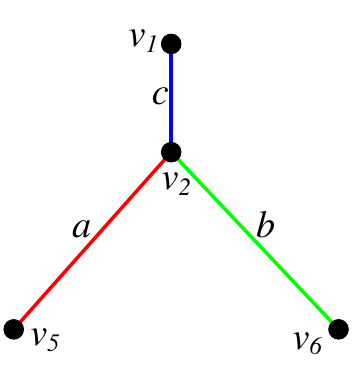}
   \label{fig_16_b}
   } \\
	\subfigure[A 3-edge coloring of $G$.]{
  \includegraphics[scale=1]{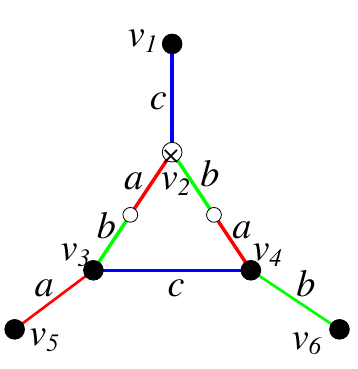}
   \label{fig_16_c}
   } \qquad
 \subfigure[A proper 3-edge coloring of $G$.]{
  \includegraphics[scale=1]{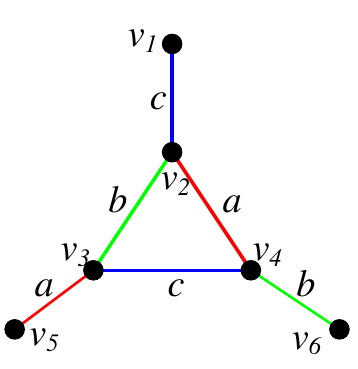}
   \label{fig_16_d}
   }
 \caption{Recursive 3-edge coloring of $G$ when girth is 3.} 
\label{fig_16}
\end{figure}
	
\begin{figure}[htbp]
 \centering

\begin{tabular}{ccc}
		  \subfigure[A quadrangle in $G$.]{
				\includegraphics[scale=0.7]{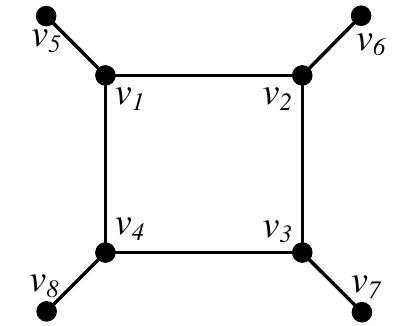}
				 \label{fig_17_a}
				 }\\
			\subfigure[A 3-edge coloring of $G'$.]{
				\includegraphics[scale=0.7]{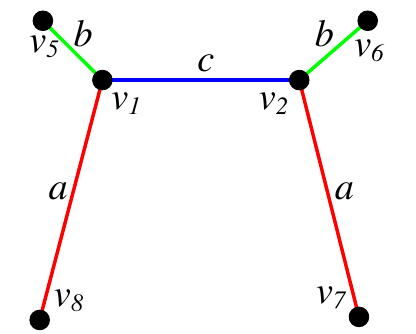}
				\label{fig_17_b}
				}&
				\subfigure[A 3-edge coloring of $G$.]{
				\includegraphics[scale=0.7]{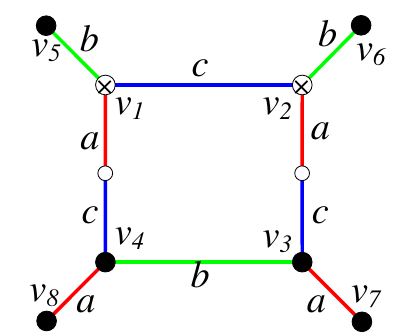}
				\label{fig_17_c}
				}&
			\subfigure[A proper 3-edge coloring of $G$.]{
				\includegraphics[scale=0.7]{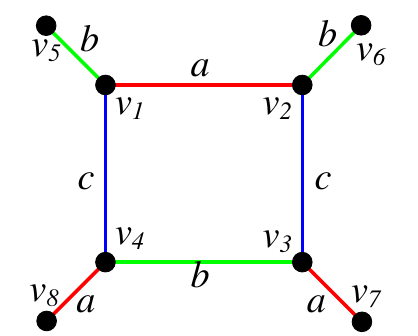}
				\label{fig_17_d}
				} \\
				\subfigure[A 3-edge coloring of $G'$.]{
				\includegraphics[scale=0.7]{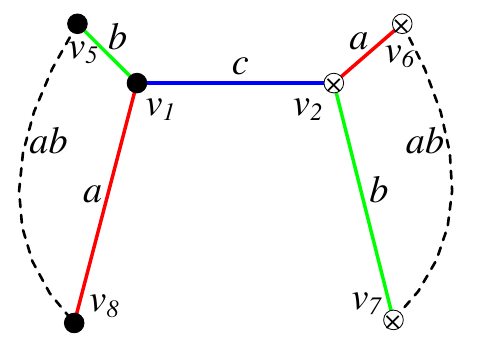}
				\label{fig_17_e}
				}&
			\subfigure[Coloring of $G'$ after negating one ($a,b$) cycle.]{
				\includegraphics[scale=0.7]{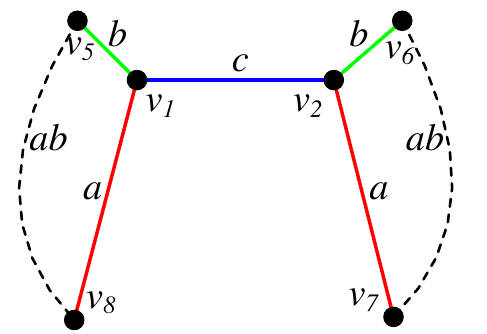}
				\label{fig_17_f}
				}&
				\subfigure[A proper 3-edge coloring of $G$.]{
				\includegraphics[scale=0.7]{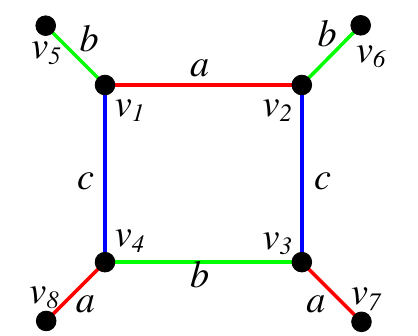}
				\label{fig_17_g}
				} \\
			\subfigure[A 3-edge coloring of $G'$.]{
				\includegraphics[scale=0.7]{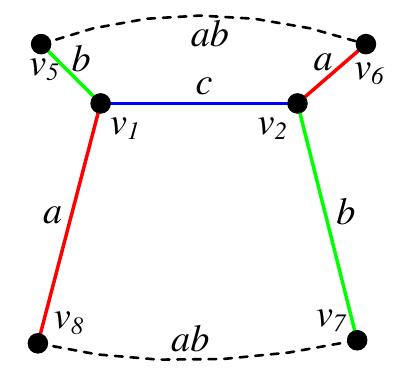}
				\label{fig_17_h}
				}&
				\subfigure[A 3-edge coloring of $G$.]{
				\includegraphics[scale=0.7]{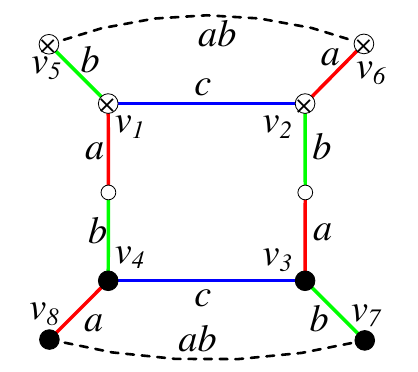}
				\label{fig_17_i}
				}&
			\subfigure[A proper 3-edge coloring of $G'$.]{
				\includegraphics[scale=0.7]{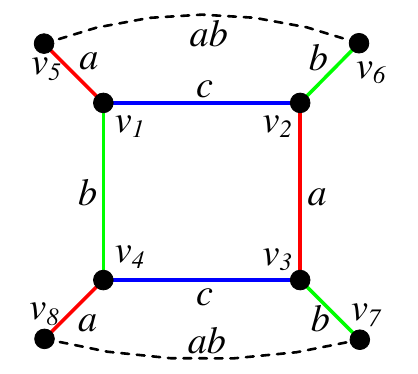}
   \label{fig_17_j}
   }
\end{tabular}
 \caption{Recursive 3-edge coloring of $G$ when girth is 4.} 
\label{fig_17}
\end{figure}

	\begin{figure}[htbp]
 \centering
 \subfigure[A pentagon in $G$.]{
  \includegraphics[scale=0.7]{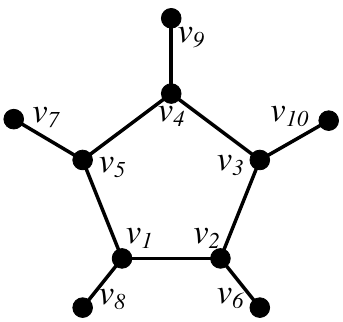}
   \label{fig_18_a}
   } \qquad
 \subfigure[A 3-edge coloring of $G'$.]{
  \includegraphics[scale=0.7]{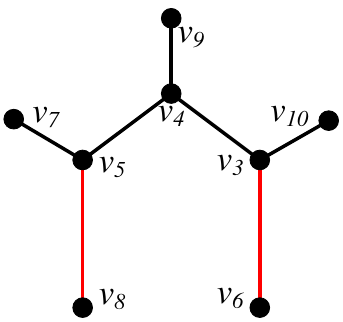}
   \label{fig_18_b}
   } \qquad
	\subfigure[A 3-edge coloring of $G$.]{
  \includegraphics[scale=0.7]{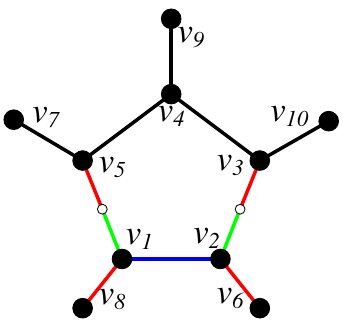}
   \label{fig_18_c}
   } \\
 \subfigure[A 3-edge coloring of $G'$.]{
  \includegraphics[scale=0.7]{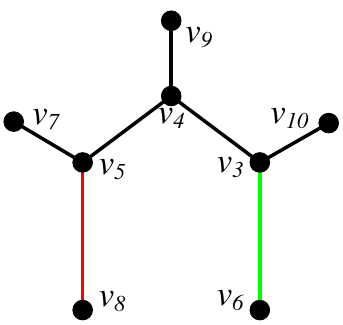}
   \label{fig_18_d}
   } \qquad
	\subfigure[A 3-edge coloring of $G$.]{
  \includegraphics[scale=0.7]{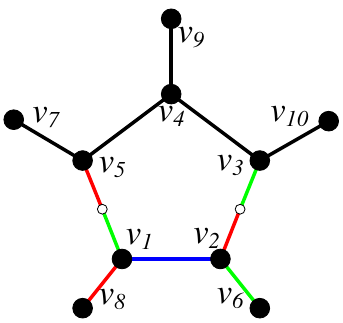}
   \label{fig_18_e}
   }
 \caption{Recursive 3-edge coloring of $G$ when girth is 5.} 
\label{fig_18}
\end{figure}
	
	\item The girth of $G$ equals 4.

	Suppose the cubic graph $G(V,E)$ with $n$ vertices has a face $F(v_1, v_2, v_3, v_4)$ with four boundary edges as shown in Fig.~\ref{fig_17_a}. We assume that edge $e_{3,4}=(v_3,v_4)$ is admissible. As before, a bridgeless cubic graph $G'$ with $n-2$ vertices can be obtained by deleting the edge $e_{3,4}$ and the two end vertices $v_3$ and $v_4$. By induction hypothesis, the cubic graph $G'$ has a 3-edge coloring. Suppose the color $c$ is assigned to edge $e_{1,2}=(v_1,v_2)$ in $G'$. Then the coloring of edge $e_{3,4}=(v_3,v_4)$ in $G$ is determined by the coloring of edge $e_{1,8}=(v_1,v_8)$ and $e_{2,7}=(v_2,v_7)$ in $G'$. We consider the following two sub-cases:

	\begin{enumerate}
		\item \label{subcase1} If edge $e_{1,8}$ and $e_{2,7}$ are colored with same color, say color $a$, as shown in Fig.~\ref{fig_17_b}. Then we can assign color $b$ to the inserted edge $e_{3,4}=(v_3,v_4)$, as shown in Fig.~\ref{fig_17_c}. The two $(a,c)$ variables created by inserting edge $e_{3,4}$ can be eliminated by a Kempe walk along the path $(v_4-v_1-v_2-v_3)$, by performing color exchanges at $v_1$ and $v_2$. We obtain a proper 3-edge coloring of $G$, as shown in Fig.~\ref{fig_17_d}, after eliminating these two $(a,c)$ variables.

 \item 	If edge $e_{1,8}$ and $e_{2,7}$ are colored with different colors, then we consider the following two circumstances:
		\begin{enumerate}
		 \item As shown in Fig.~\ref{fig_17_e}, the two $(a,b)$ paths $(v_5-v_1-v_8)$ and $(v_6-v_2-v_7)$ are contained in different $(a,b)$ cycles. Negate one of the $(a,b)$ cycle in $G'$, say the one that contains the edge $e_{2,7}$. As a result, the two edges $e_{1,8}$ and $e_{2,7}$ have the same color $a$ as shown in Fig.~\ref{fig_17_f}. Then we can easily derive a 3-edge coloring of $G$, as shown in Fig.~\ref{fig_17_g}, in a way similar to sub-case \ref{subcase1} described above.

		\item As shown in Fig.~\ref{fig_17_h}, the two $(a,b)$ paths $(v_5-v_1-v_8)$ and $(v_6-v_2-v_7)$ are contained in the same $(a,b)$ cycle. We first insert edge $e_{3,4}=(v_3,v_4)$ into $G'$. The inserted $(c, c)$ edge $e_{3,4}$ introduces two $(a,b)$ variables at the two ends to make the coloring of $G$ consistent, as shown in Fig.~\ref{fig_17_i}. Since these two variables are contained in the same $(a,b)$ cycle, they can be easily eliminated by a Kempe walk, as shown in Fig.~\ref{fig_17_i}. The resulting 3-edge coloring of $G$ is exhibited in Fig.~\ref{fig_17_j}.
		\end{enumerate}

\end{enumerate}

	\item The girth of G equals 5.

	Suppose the cubic graph $G(V,E)$ with $n$ vertices has a face $F(v_1, v_2, v_3, v_4, v_5)$ with five boundary edges as shown in Fig.~\ref{fig_18_a}. We assume that edge $e_{1,2}=(v_1,v_2)$ is admissible. We can obtain a graph $G'$ with $n-2$ nodes by deleting edge $e_{1,2}$ and smoothing out the two end vertices $v_1$ and $v_2$, as shown in Fig.~\ref{fig_18_b}. 
	
By induction hypothesis, the cubic graph $G'$ has an 3-edge coloring. The two edges $e_{5,8}$ and $e_{3,6}$ in $G'$ may have the same color, say color $a$, as shown in Fig.~\ref{fig_18_b}, or different colors, say color $a$ and $b$, as shown in Fig.~\ref{fig_18_d}. In both cases, two $(a,b)$ variables will be created if we insert edge $e_{1,2}$ back into $G'$, as shown in Fig.~\ref{fig_18_c} and \ref{fig_18_e}, respectively. If the two variables are contained in the same $(a,b)$ cycle, then they can be eliminated by a simple Kempe walk. Otherwise, graph $G$ has a Petersen configuration $P(G)$ if these two variables are contained in different $(a,b)$ cycles. Graph $G$ is 3-edge colorable because the Petersen configuration $P(G)$ is reducible according to the reducibility postulate.

\end{enumerate}

\end{proof}

{\parindent0pt \textbf{Remark:}}
The 4CT is therefore established according to Tait's equivalent formulation mentioned in Section \ref{sec1}. It should be noted that because the reducibility postulate requires that the Petersen configuration can be reduced by a single essential cycle, the 4CT does not imply that every Petersen configuration has this property. Therefore, the existing computer-assisted proof of 4CT does not apply to the proof of reducibility postulate.

\section{Comparisons and Discussions}
\label{sec7}
The reducibility postulate of the Petersen configuration in graph theory and the parallel postulate in Euclidean geometry share some common characteristics of the two-dimensional plane. A comparison between these two propositions is described as follows.

\begin{itemize}
\item \textbf{\textit{Invariants of the plane.}}

In the two-dimensional plane, the well-known geometric invariant that the sum of the angles in every triangle equals $\pi$ is an immediate consequence of the parallel postulate. Similarly, in the last section, we proved that the reducibility postulate of the Petersen configuration implies that the chromatic index of bridgeless cubic planar graphs equals 3, or equivalently, the minimum number of colors to color a geographical map is 4, which is a topological invariant of the plane. 

\item \textbf{\textit{Solvability conditions of two equations in the plane.}}

In Euclid’s Elements, postulate 5, the original version of the parallel postulate, is stated as follows: 
\begin{quote}
\textit{If a straight line falling on two straight lines makes the interior angles on the same side less than two right angles, the two straight lines, if produced indefinitely, meet on that side on which are the angles less than the two right angles.}
\end{quote}

Consider each line as a linear equation. The parallel postulate is actually the solvability condition of two linear equations in the plane, which is somewhat analogous to that of the reducibility postulate. In a bridgeless cubic planar graph $G$, the reducibility postulate claims that every Petersen configuration $P(G)$ of $G$ is reducible, or equivalently, the two $(a,b)$ variables of $P(G)$ can be connected by an $(a,b)$ Kempe path in some state of $P(G)$, via negating a resolution cycle. The common characteristics of the parallel postulate and the reducibility postulate are listed in Table~\ref{tab2}.

\begin{table}[ht]
\centering
\caption{A comparison between parallel postulate and postulate of reducibility} 
\begin{tabular}{| c | c |}
\hline
\textbf{Parallel Postulate} 
 & \textbf{Reducibility Postulate} 
 \\ 
\hline
In the plane & In the plane  \\
\hline
Two lines intersected by a third line & Two locking cycles joined by an essential resolution cycle  \\
\hline
Sum of inner angles $<\pi$ & Bounded by pentagon \\
\hline
\end{tabular}
\label{tab2}
\end{table}

\end{itemize}

The analogy between these two solvability propositions in the plane is illustrated in Fig.~\ref{fig_19}. The parallel postulate asserts the solvability condition of two lines by using a third line, as highlighted in Fig.~\ref{fig_19_a}, that intersects both lines, while the reducibility postulate asserts the solvability of two odd $(a,b)$ cycles by using an essential $(a,c)$ resolution cycle, as highlighted in Fig.~\ref{fig_19_b}, that "intersects" both locking cycles. The two $(a,b)$ locking cycles and an essential $(a,c)$ resolution cycle that compose the Petersen configuration are shown in Fig.~\ref{fig_19_c} and Fig.~\ref{fig_19_d}, respectively.

Furthermore, the counterpart of the condition $\theta_1+\theta_2<\pi$ is that both $(a,b)$ locking cycles must contain the edges of a face with less than six border edges. Note that the boundary of the entire configuration shown in Fig.~\ref{fig_19_b} is a pentagon. These two similar bounding conditions represent some extruding constraints that ensure the existence of a solution in the plane.

\begin{figure}[hbtp]
 \centering
 \subfigure[Parallel postulate.]{
  \includegraphics[scale=0.7]{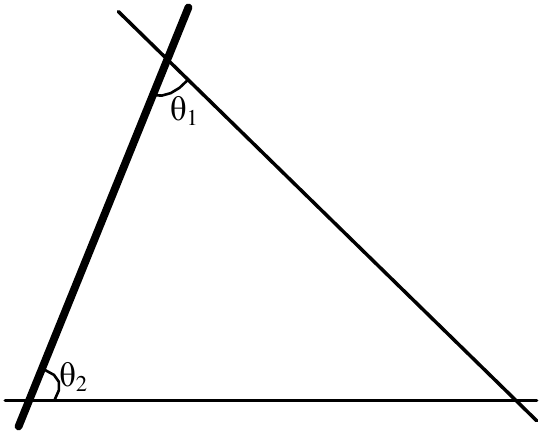}
   \label{fig_19_a}
   } \qquad
 \subfigure[Reducibility postulate.]{
  \includegraphics[scale=0.7]{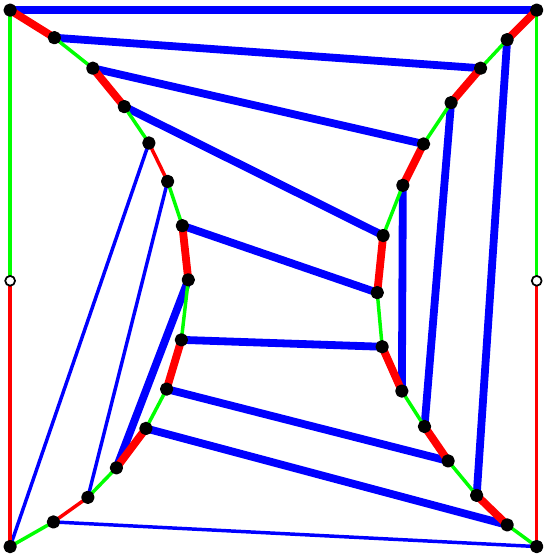}
   \label{fig_19_b}
   } \\
	\subfigure[Two $(a,b)$ locking cycles.]{
  \includegraphics[scale=0.7]{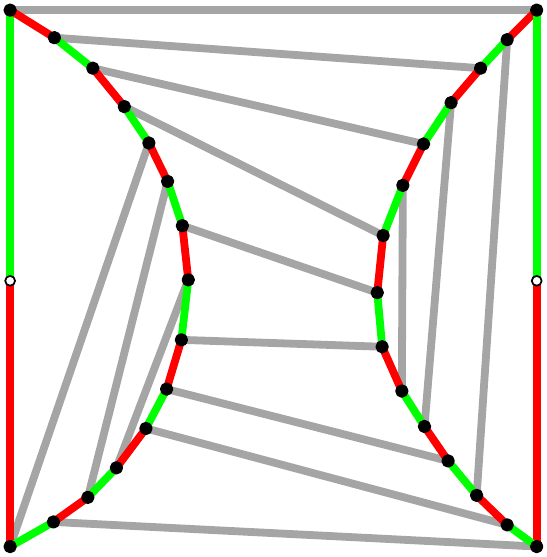}
   \label{fig_19_c}
   } \qquad
	\subfigure[$(a,c)$ essential resolution cycle.]{
  \includegraphics[scale=0.7]{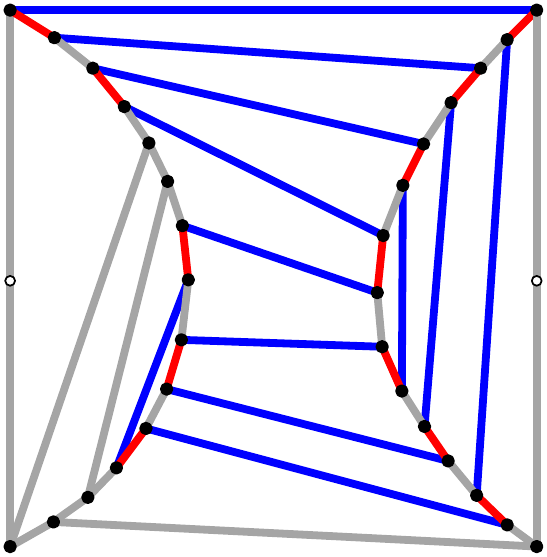}
   \label{fig_19_d}
   } 
 \caption{Analogy between parallel postulate and reducibility postulate.}  
\label{fig_19}
\end{figure}

On the other hand, the Petersen configuration may be irreducible if the cubic graph is non-planar. As we illustrated in Fig.~\ref{fig_9} before, all configurations of the Petersen graph are irreducible and they are isomorphic to each other. The above analogy between the parallel postulate and the reducibility postulate not only reveals the significance of the concept of complex coloring in graph theory, but also hints that a logical proof of the proposition is most likely impossible, the same as parallel postulate.

\section{Conclusions}
\label{sec8}
As we mentioned before, the configuration of a bridgeless cubic graph $G$ is uniquely determined by Petersen's perfect matching. Lov\'{a}sz and Plummer conjectured that the number of perfect matchings contained in $G$ is exponential to the number of vertices of $G$ in \cite{Lovasz1986}. The conjecture was settled by Esperet et al. in \cite{Esperet20111646}. The identification of snark $G$ with $n$ vertices is determined by the number of non-isomorphic configurations, denoted as $\theta(G)$, in the closed set of irreducible configurations of $G$, and the number of variables in each configuration, called \textbf{\textit{oddness}} and denoted as $\kappa(G)$ \cite{Huck1995119}. It has been shown in \cite{Steffen2004191} that the oddness $\kappa(G)$ is unbounded. Presumably, the complexity of deciding whether a cubic graph is a snark should be at least on the order of $O(\theta n^\kappa)$, even excluding the testing of isomorphism. Therefore, it is not feasible to deduce an efficient deterministic algorithm to identify snarks. On the other hand, as snarks are clearly specified, the development of some efficient randomized algorithms is possible in the future.

In respect to computation complexity, the postulate of reducibility is consistent with the computer-assisted proof of the 4CT. Despite the fact that it is NP-complete to determine the chromatic index of an arbitrary cubic graph, this problem for planar graphs can be solved in polynomial time because the computer-assisted proof of the 4CT actually gives a quadratic algorithm for map coloring, as described in \cite{Robertson1996}. The proof of the 3-edge coloring theorem presented in Section \ref{sec6} also implies a polynomial time algorithm. The consistency of the complexities of these two drastically different approaches suggests that the reduciblity postulate of the Petersen configuration should hold in every bridgeless cubic planar graph.

The complexities for edge coloring of bridgeless cubic graphs are classified in Fig.~\ref{fig_20}, based on the assumption that the postulate of reducibility is valid. Suppose, by contrast, that some Petersen configurations of bridgeless cubic planar graphs are irreducible, such that the postulate of reducibility is invalid. Then the complexity of edge coloring of planar cubic graphs would also be in the class of NP-complete, the same as non-planar cases, because it would involve searching for a reducible state in the entire space of all configurations. This scenario conflicts with the quadratic algorithm for map coloring derived from the computer-assisted proof of the 4CT. 

\begin{figure}[htbp]
 \centering
  \includegraphics[scale=0.9]{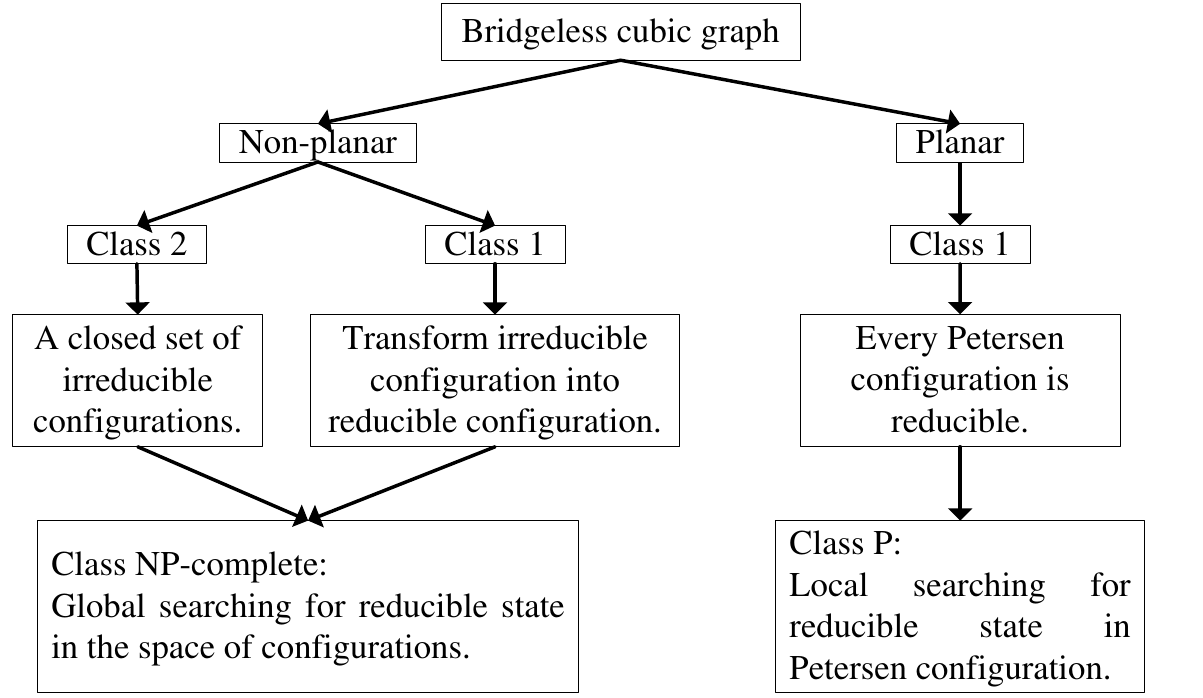}
 \caption{The classification of bridgeless cubic graphs.} 
\label{fig_20}
\end{figure}

Another NP-complete problem that could be tackled by complex color exchanges is finding Hamiltonian cycles \cite{papadimitriou1995computational}. A simple graph $G$ may have more than one proper color configurations, and a configuration can be transformed into other configurations by negating maximal two-colored Tait cycles. A configuration is Hamiltonian if it contains a two-colored Hamiltonian cycle. By random walks on the entire space of configurations, a solution of a given graph $G$ could be reached if $G$ is Hamiltonian. This is only an example to show that the application of complex coloring to solve some hard combinatorial problems could be a challenging research topic in the future.

\section*{Acknowledgment}
Tony T. Lee was supported by the National Science Foundation of China under Grant 61172065. Qingqi Shi was supported by the Hong Kong RGC Earmarked Grant CUHK414012.

\bibliography{bib4ct}
\bibliographystyle{aomalpha}

\appendix
%Appendix A
\section{Contraction of Snarks to the Petersen graph}
\label{appdx_A}

The contraction of a snark with two $(a,b)$ variables to the Petersen graph can be achieved by the following procedure:

\begin{enumerate}
	\item   Delete all $(a,a)$ edges that are not contained in the two odd $(a, b)$ cycles, and smooth out all degree two vertices on the $(b,c)$ chains, which become $(c,c)$ edges.
	\item   Delete all internal chords of the two odd $(a, b)$ cycles, and smooth out all degree two vertices on these two cycles.
	\item   Find a subset of five externals chords, the $(c,c)$ edges that connect the two $(a, b)$ cycles, that forms the same connection pattern as the five external chords in the configuration $T(G_P)$ of the Petersen graph, as shown in Fig.~\ref{fig_9_a}.
	
\end{enumerate}

In the configuration described above, the term external chord refers to the $(c,c)$ edge that connects the two disjoint $(a,b)$ cycles, and the term internal chord refers to the $(c,c)$ edge with both ends terminated on the same $(a,b)$ cycle. The contraction of several snarks is shown in Fig.~\ref{fig_21} through Fig.~\ref{fig_24}.

\begin{figure}[htbp]
 \centering
 \subfigure[Flower snark.]{
  \includegraphics[scale=0.6]{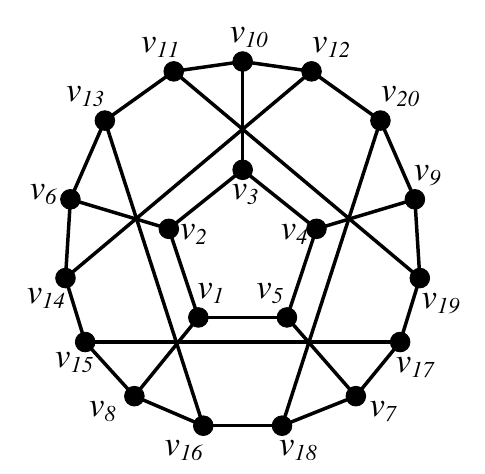}
   \label{fig_21_a}
   }  \quad
 \subfigure[A configuration of Flower snark.]{
  \includegraphics[scale=0.6]{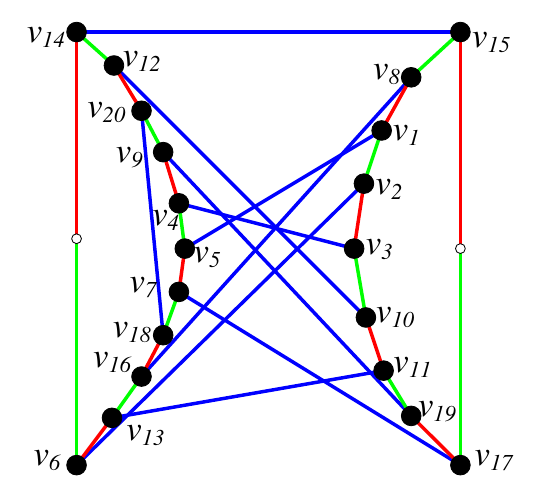}
   \label{fig_21_b}
   } \quad
	\subfigure[Find 5 chords matched with $T(G_p)$ of the Petersen graph.]{
  \includegraphics[scale=0.6]{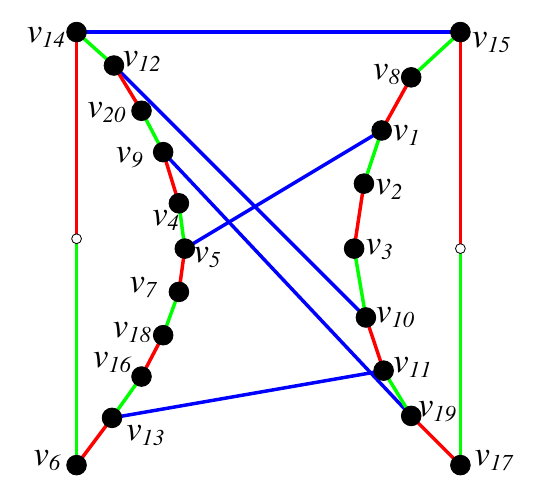}
   \label{fig_21_c}
   } \\
 \subfigure[Reduce to the Petersen graph.]{
  \includegraphics[scale=0.6]{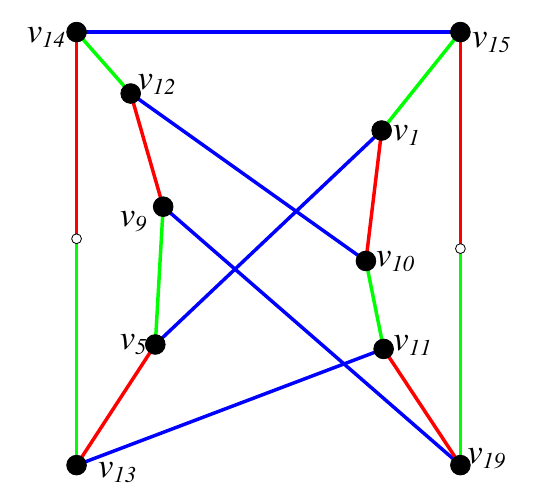}
   \label{fig_21_d}
   } \quad
	\subfigure[A subdivision of the Petersen graph embedded in Flower snark.]{
  \includegraphics[scale=0.6]{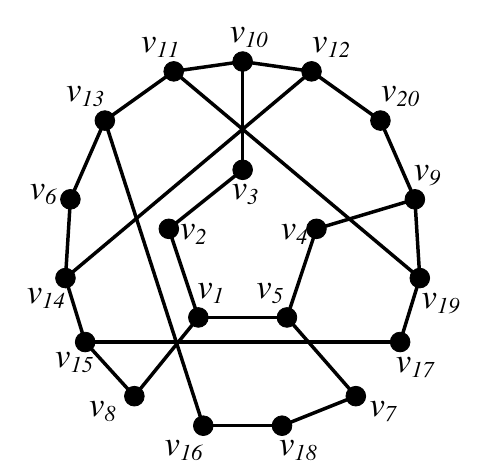}
   \label{fig_21_e}
	}
 \caption{Contraction of Flower snark to the Petersen graph.}
\label{fig_21}
\end{figure}

\begin{figure}[htbp]
 \centering
 \subfigure[Loupekine's first snark $L_1$.]{
  \includegraphics[scale=0.6]{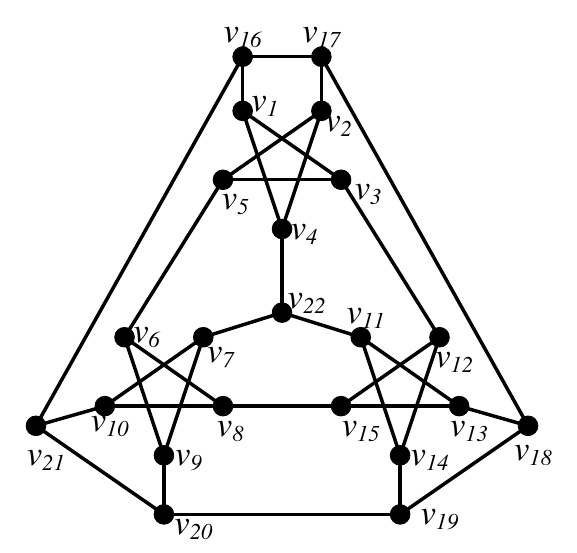}
   \label{fig_22_a}
   } \quad
 \subfigure[A configuration $T(L_1)$ of $L_1$.]{
  \includegraphics[scale=0.6]{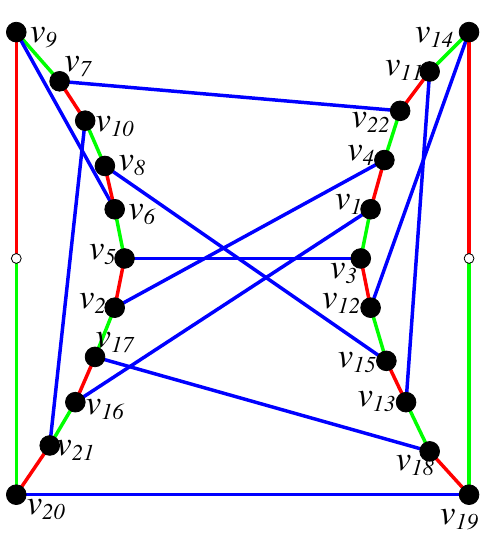}
   \label{fig_22_b}
   } \quad
	\subfigure[Remove internal chord of $L_1$.]{
  \includegraphics[scale=0.6]{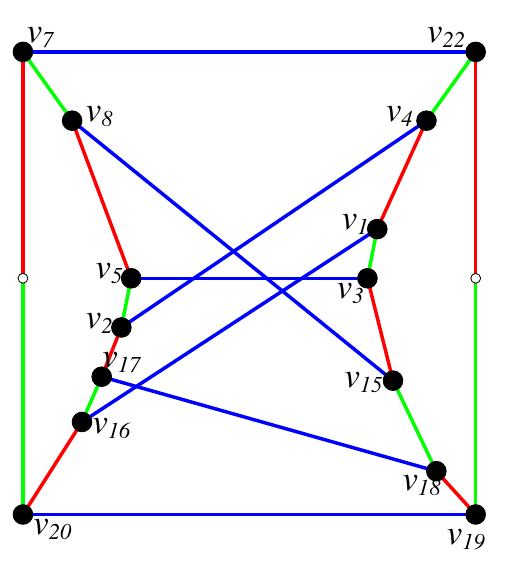}
   \label{fig_22_c}
   } \\
 \subfigure[Find 5 chords matched with $T(G_p)$ of the Petersen graph.]{
  \includegraphics[scale=0.6]{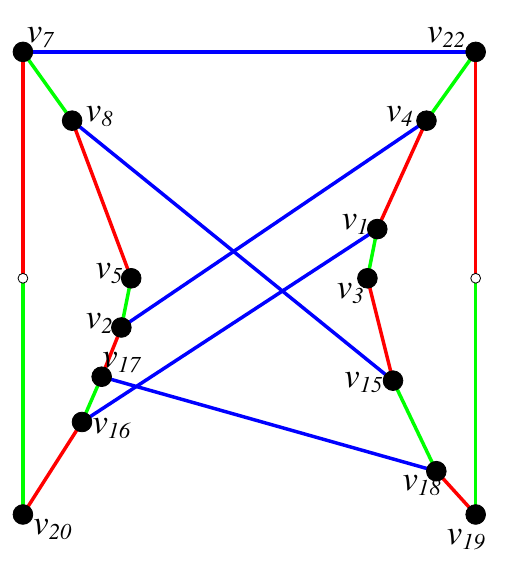}
   \label{fig_22_d}
   } \quad
	\subfigure[Reduce to the Petersen graph.]{
  \includegraphics[scale=0.6]{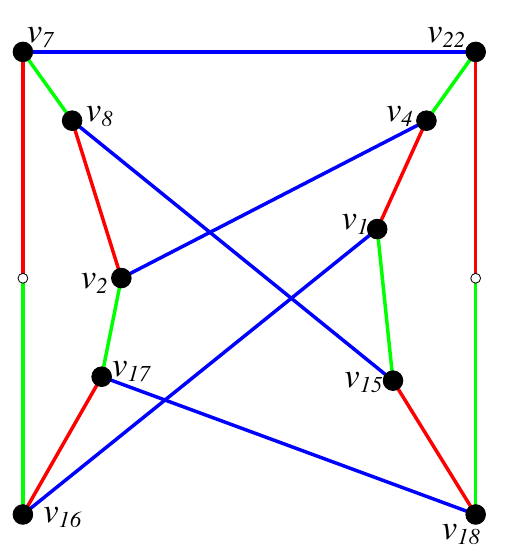}
   \label{fig_22_e}
   } \quad
 \subfigure[A subdivision of the Petersen graph embedded in Loupekine's first snark.]{
  \includegraphics[scale=0.6]{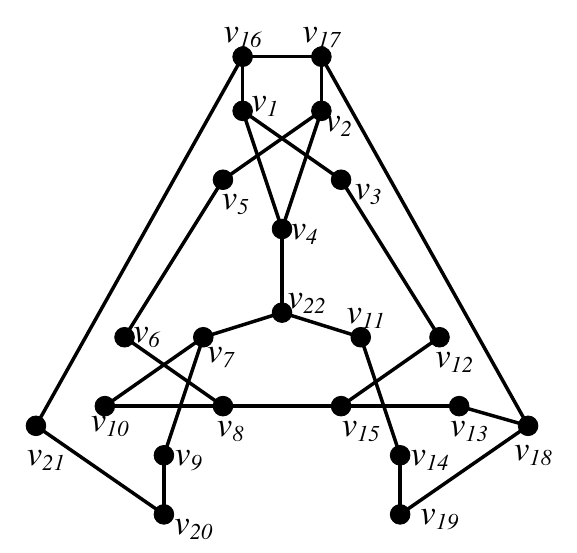}
   \label{fig_22_f}
   }
 \caption{Contraction of Loupekine's first snark to the Petersen graph.} 
\label{fig_22}
\end{figure}

\begin{figure}[htbp]
 \centering
 \subfigure[Loupekine's second snark $L_2$.]{
  \includegraphics[scale=0.6]{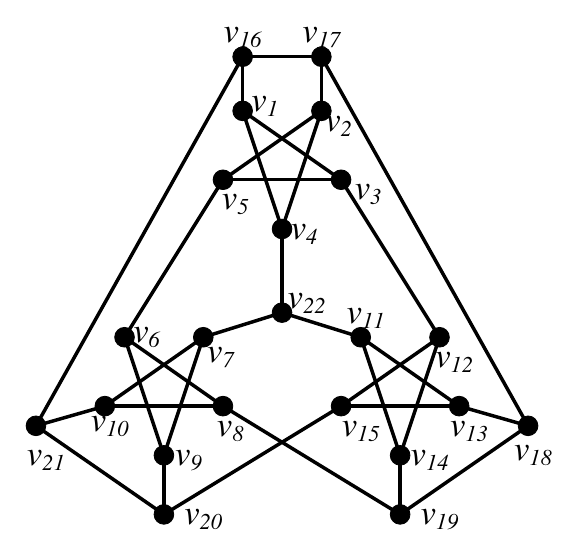}
   \label{fig_23_a}
   } \quad
 \subfigure[A configuration $T(L_2)$ of $L_2$.]{
  \includegraphics[scale=0.6]{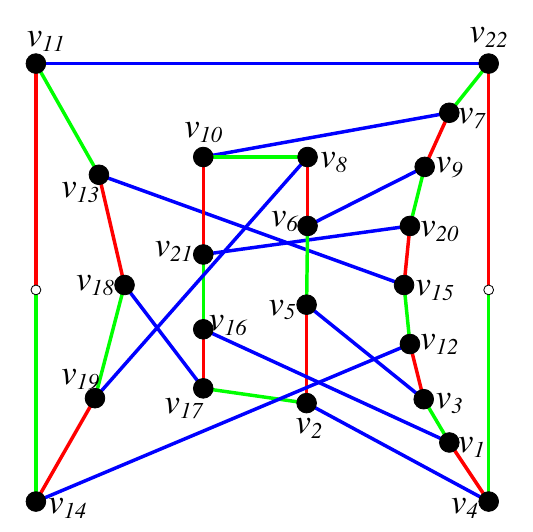}
   \label{fig_23_b}
   } \quad
	\subfigure[Delete $(a,a)$ edges in even $(a,b)$ cycles.]{
  \includegraphics[scale=0.6]{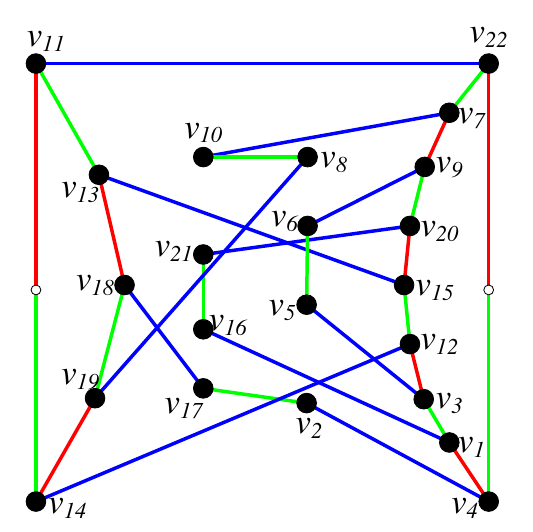}
   \label{fig_23_c}
   } \\
 \subfigure[Smooth vertices of degree 2.]{
  \includegraphics[scale=0.6]{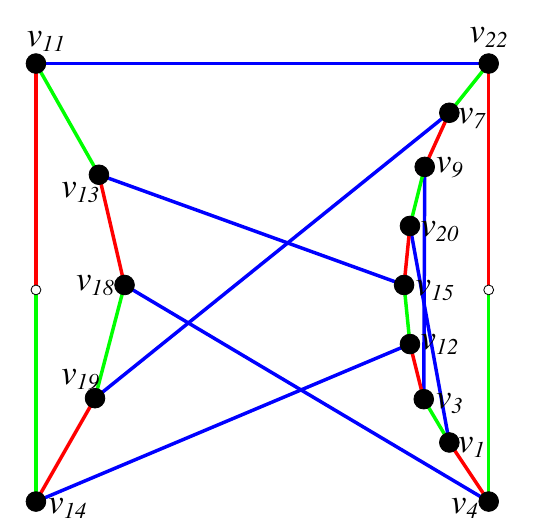}
   \label{fig_23_d}
   } \quad
	\subfigure[Delete internal chord and smooth vertices of degree 2.]{
  \includegraphics[scale=0.6]{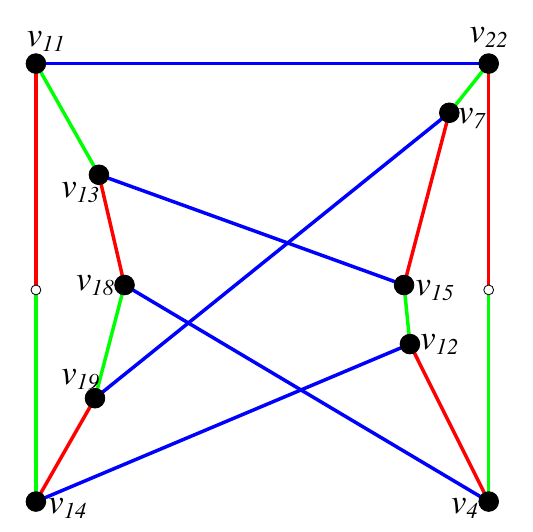}
   \label{fig_23_e}
   } \quad
 \subfigure[A subdivision of Petersen graph embedded in Loupekine's second snark.]{
  \includegraphics[scale=0.6]{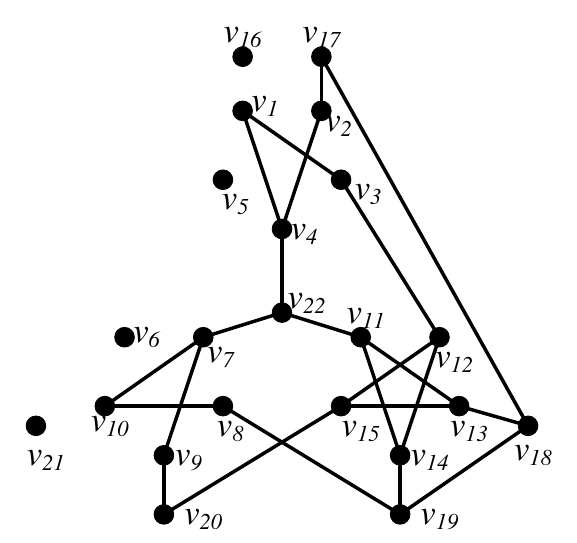}
   \label{fig_23_f}
   }
 \caption{Contraction of Loupekine's second snark to the Petersen graph.} 
\label{fig_23}
\end{figure}

\begin{figure}[htbp]
 \centering
 \subfigure[Double star snark.]{
  \includegraphics[scale=0.6]{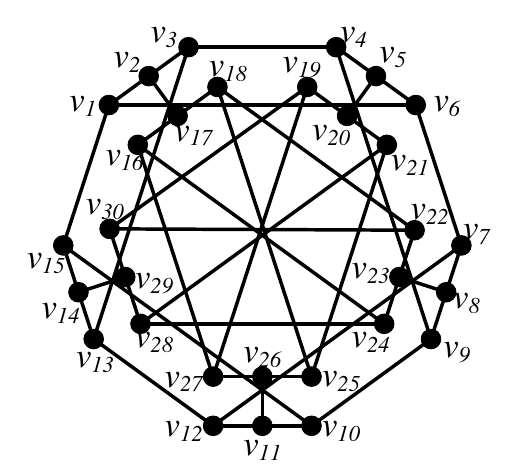}
   \label{fig_24_a}
   } \quad
 \subfigure[A configuration of double star snark.]{
  \includegraphics[scale=0.6]{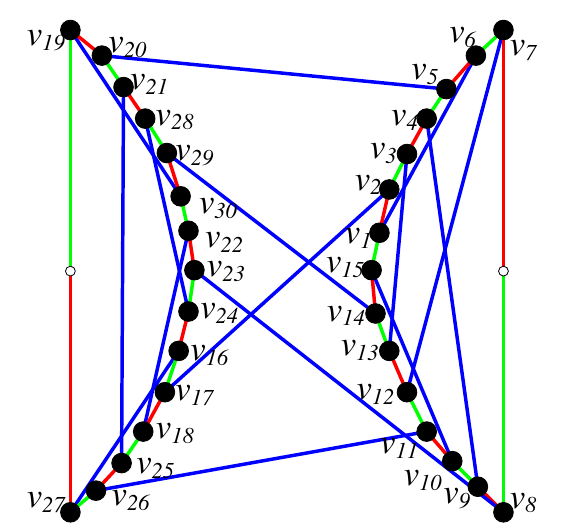}
   \label{fig_24_b}
   } \quad
	\subfigure[Delete internal chords.]{
  \includegraphics[scale=0.6]{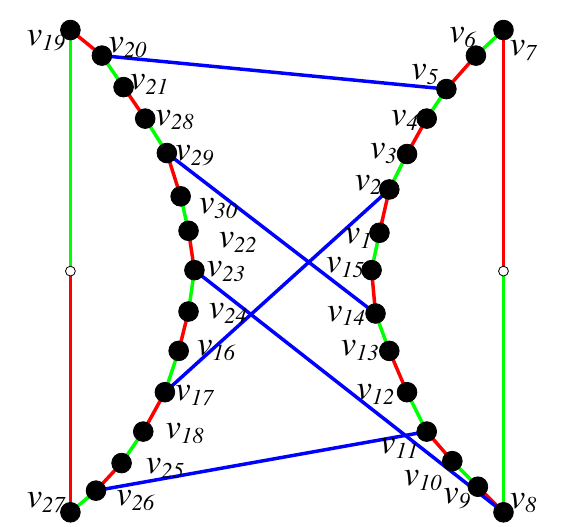}
   \label{fig_24_c}
   } \\
 \subfigure[Reduce to the Petersen graph.]{
  \includegraphics[scale=0.6]{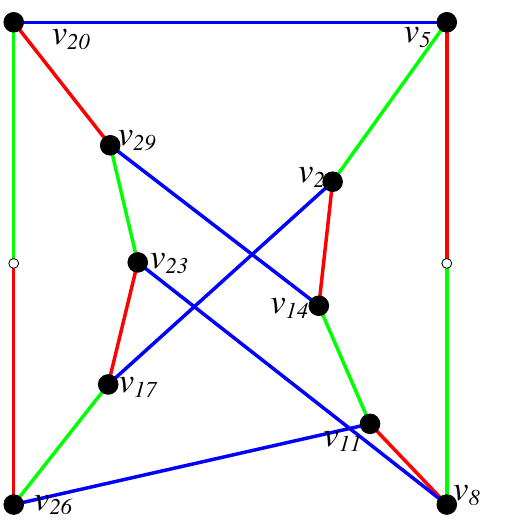}
   \label{fig_24_d}
   } \quad
	\subfigure[A subdivision of the Petersen graph embedded in double star snark.]{
  \includegraphics[scale=0.6]{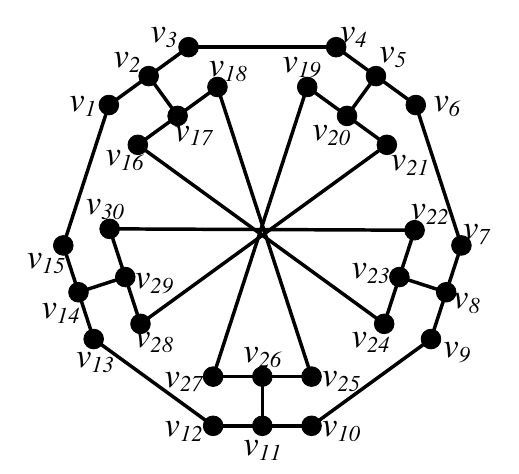}
   \label{fig_24_e}
   }
 
 \caption{Contraction of double star snark to the Petersen graph.}\label{fig_24}
\end{figure}
\end{document}